\documentclass[11pt]{article}

\usepackage{bm, commath}
\usepackage{natbib}
\usepackage{caption}
\usepackage{graphicx}
\usepackage{subfigure}
\usepackage{amsmath, amsfonts, amsthm}
\usepackage{float}
\usepackage{booktabs,siunitx}
\usepackage{url}
\usepackage{multirow}

\usepackage{listings}
\usepackage[margin=1in]{geometry}
\usepackage{authblk}
\usepackage[ruled]{algorithm2e}
\SetKwInput{KwParam}{Parameter}
\SetAlgoCaptionLayout{centerline}

\usepackage{sectsty}
\usepackage[onehalfspacing]{setspace}
\usepackage[usenames,dvipsnames]{color}
\usepackage[utf8]{inputenc}
\usepackage{float}
\usepackage{tikz}
\usetikzlibrary{shapes,decorations,arrows,calc,arrows.meta,fit,positioning}
\tikzset{
    -Latex,auto,node distance =1 cm and 1 cm,semithick,
    state/.style ={ellipse, draw, minimum width = 0.7 cm},
    point/.style = {circle, draw, inner sep=0.04cm,fill,node contents={}},
    bidirected/.style={Latex-Latex,dashed},
    el/.style = {inner sep=2pt, align=left, sloped}
}

\newcommand{\indep}{\rotatebox[origin=c]{90}{$\models$}}

\newtheorem{definition}{Definition}

\newtheorem{remark}{Remark}
\newtheorem{example}{Example}

\newtheorem{proposition}{Proposition}
\newtheorem{condition}{Condition}
\newtheorem{lemma}{Lemma}

\usepackage{mathtools}

\usepackage{xcolor}

\makeatletter
\renewcommand{\algocf@captiontext}[2]{#1\algocf@typo. \AlCapFnt{}#2} 
\def\@algocf@capt@plain{top}
\renewcommand{\algocf@makecaption}[2]{%
  \addtolength{\hsize}{\algomargin}%
  \sbox\@tempboxa{\algocf@captiontext{#1}{#2}}%
  \ifdim\wd\@tempboxa >\hsize
  \hskip .5\algomargin%
  \parbox[t]{\hsize}{\algocf@captiontext{#1}{#2}}
  \else%
  \global\@minipagefalse%
  \hbox to\hsize{\box\@tempboxa}
  \fi%
  \addtolength{\hsize}{-\algomargin}%
}
\makeatother

\begin{document}

\sectionfont{\bfseries\large\sffamily}%

\subsectionfont{\bfseries\sffamily\normalsize}%

\begin{center}
\noindent
{\sffamily\bfseries\LARGE
Bridging preference-based instrumental variable studies and cluster-randomized encouragement experiments: study design, noncompliance, and average cluster effect ratio}%

\noindent
$\textsf{Bo Zhang}^{1}$, $\textsf{Siyu Heng}^{2}$, $\textsf{Emily J.\ Mackay}^3$, and $\textsf{Ting Ye}^1$
\end{center}

\begin{center}
$^1$Department of Statistics, The Wharton School, University of Pennsylvania, Philadelphia, Pennsylvania, U.S.A.
U.S.A. Correspondence: bozhan@wharton.upenn.edu\\
$^2$Graduate Group in Applied Mathematics and Computational Science, School of Arts and Sciences, University of Pennsylvania, Philadelphia, Pennsylvania, U.S.A.\\
$^3$Department of Anesthesiology and Critical Care, Perelman School of Medicine, University of Pennsylvania, Philadelphia, Pennsylvania, U.S.A.    
\end{center}

\vspace{0.5 cm}
\noindent
\textsf{{\bf Abstract}: Instrumental variable methods are widely used in medical and social science research to draw causal conclusions when the treatment and outcome are confounded by unmeasured confounding variables. One important feature of such studies is that the instrumental variable is often applied at the cluster level, e.g., hospitals' or physicians' preference for a certain treatment where each hospital or physician naturally defines a cluster. This paper proposes to embed such observational instrumental variable data into a cluster-randomized encouragement experiment using statistical matching. Potential outcomes and causal assumptions underpinning the design are formalized and examined. Testing procedures for two commonly-used estimands, Fisher's sharp null hypothesis and the pooled effect ratio, are extended to the current setting. We then introduce a novel cluster-heterogeneous proportional treatment effect model and the relevant estimand: the average cluster effect ratio. This new estimand is advantageous over the structural parameter in a constant proportional treatment effect model in that it allows treatment heterogeneity, and is advantageous over the pooled effect ratio estimand in that it is immune to Simpson's paradox. We develop an asymptotically valid randomization-based testing procedure for this new estimand based on solving a mixed integer quadratically-constrained optimization problem. The proposed design and inferential methods are applied to a study of the effect of using transesophageal echocardiography during CABG surgery on patients' 30-day mortality rate.}%

\vspace{0.3 cm}
\noindent
\textsf{{\bf Keywords}: Heterogeneous treatment effect; Instrumental variable; Mixed-integer programming; \\ Randomization-based inference; Statistical matching}

\newpage

\section{Introduction}
\label{sec: introduction}
\subsection{Instrumental variable, encouragement, and cluster-randomized encouragement design}
 One central goal in social science and medical research is to evaluate a program or treatment assignment. When a randomized controlled trial is not an option due to ethical or practical considerations, instrumental variable analysis provides an attractive alternative. An instrumental variable can be best thought of as a haphazard encouragement to take a treatment. In many practical situations, although treatment assignment cannot be randomized, encouragement can be, and this motivates the so-called randomized encouragement design (\citealp{holland1988causal}). An important feature of the randomized encouragement design is that the encouragement is often applied at the cluster level (e.g., villages, hospitals, or physicians). A cluster-randomized encouragement design refers to a research design where entire clusters are randomly assigned to an encouragement to take a particular treatment, but individuals within each cluster are allowed to choose individually whether to receive the treatment or not. A cluster-randomized trial with individual noncompliance can be thought of as an ideal prototype of a cluster-randomized encouragement design. Cluster-randomized encouragement designs are widely applicable in practice (\citealp{sommer1991estimating,frangakis2002clustered,imai2009essential}). For instance, \citet{sommer1991estimating} reported a study where villages in Indonesia were randomized to receive vitamin A supplements for newborns, but not all infants who were assigned vitamin A supplements actually received them. 

Many existing cluster-level instrumental variables are not randomized as in an experiment. For example, one of the most popular instrumental variables in empirical studies is physicians' preference for a particular treatment (\citealp{garabedian2014potential}). Physicians' preference as an instrument naturally falls into the framework of a cluster-level instrumental variable, because each physician defines a natural cluster of patients (i.e., those treated by the physician), and patient-level noncompliance is commonplace. Unlike cluster-randomized trials with noncompliance, where experimenters can manipulate and fully randomize the encouragement, physicians' preference may be associated with the outcome because there are ``IV-outcome'' confounders (\citealp{garabedian2014potential}). For example, physicians' preference may be associated with patient volume and institute policies, both of which are likely to be associated with clinical outcomes and thus invalidate a naive instrumental variable analysis that assumes the preference-based instrument is valid without conditional on relevant covariates. 
 
 \subsection{Our contribution}
 Authors from different disciplines have proposed methods to study cluster-randomized encouragement designs under different causal frameworks. See Supplementary Material A for a detailed literature review. To the best of our knowledge, this article proposes the first matching-based study design approach that embeds observational instrumental variable data with cluster-level continuous instruments into a cluster-randomized encouragement experiment. We adapt optimal nonbipartite matching techniques (\citealp{lu2001matching, baiocchi2010building,lu2011optimal}) to pairing similar clusters with markedly different cluster-level continuous instruments (e.g., physicians who had treated similar patients but with different preference for the treatment) without first dichotomizing the continuous instrument, and consider statistical inference under individual noncompliance. We have three objectives. In Section \ref{sec: notation and assumption}, we examine and clarify the potential outcomes and causal assumptions underpinning our study design, and discuss in detail motivations and practical advantages of embedding observational instrumental variable data, when appropriate, into a cluster-randomized encouragement experiment as opposed to a non-clustered, individual-level, encouragement experiment. In Section \ref{sec: random-based inference}, we review and generalize inferential procedures testing Fisher's sharp null hypothesis and the effect ratio, an estimand that allows for treatment heterogeneity (\citealp{imai2009essential}, \citealp{baiocchi2010building,kang2016full, kang2018estimation}), to the current setting. The generalized effect ratio estimand is referred to as the pooled effect ratio (PER). In Section \ref{sec: ACER}, we propose to largely relax the previously studied constant proportional treatment effect model (\citealp{small2008war, small2008randomization}) by considering a version that allows cluster-heterogeneous proportional treatment effect. The relevant estimand, known as the average cluster effect ratio (ACER), has an advantage over the pooled effect ratio that it is immune to Simpson's paradox. We develop an asymptotically valid randomization-based testing procedure for this new estimand based on solving a mixed integer quadratically-constrained optimization problem. We apply the proposed design and testing procedures to studying the effect of transesophageal echocardiography (TEE) monitoring during coronary artery bypass graft (CABG) surgery on patients' 30-day mortality rate using the U.S Medicare and Medicaid claims data in Section \ref{sec: application}. We implement the proposed method in the \textsf{R} package \textsf{ivdesign}.

\subsection{Application: effect of TEE monitoring during CABG surgery on 30-day mortality rate}
\label{subsec: application in intro}
Coronary artery bypass graft (henceforth CABG) surgery is the most widely performed adult cardiac surgery, accounting for over half of the $300,000$ cardiac surgeries performed in the U.S. each year (\citealp{Database_2018}). Transesophageal echocardiography (henceforth TEE) is an ultrasound-based, cardiac imaging modality, frequently used in cardiac surgery for hemodynamic monitoring and management of complications (\citealp{mackay2020transesophageal}). 

Table~\ref{tbl: balance table} presented the covariate balance of $204$ pairs of $2$ surgeons from the United States who had performed at least $30$ isolated CABG surgeries from 2013 to 2015. Similar surgeons with markedly different preference for using TEE were paired together in an optimal way using
nonbipartite matching, a flexible statistical matching technique that handles continuous exposure without first dichotomizing it (\citealp{lu2001matching}; \citealp{baiocchi2010building}; \citealp{lu2011optimal}). These $204 \times 2 = 408$ surgeons had treated a total of $18,748$ Medicare patients. We defined a surgeon's preference for TEE as the fraction of CABG surgeries performed with TEE monitoring. Two surgeons were judged similar if they treated similar patient population and practiced in similar hospitals. We are interested in answering the following clinical question with the data: Does TEE monitoring help reduce 30-day mortality rate for patients undergoing isolated CABG surgery, and if it does, then by how much? Further details on the application can be found in Supplementary Material F.

\begin{table}[ht]
\small
\centering
\caption{Covariate balance after matching. We formed $204$ pairs of $2$ surgeons from the United States. After matching, absolute standardized differences (absolute value of difference in means divided by the pooled standard deviation of two groups) of all covariates are less than $0.2$, and p-values obtained from two-sample t-tests are all greater than $0.1$, suggesting no systematic difference between two matched groups. Low-TEE-preference surgeons are referred to as control surgeons and high-TEE-preference encouraged surgeons.}
\begin{tabular}{lcccc}
& \multirow{3}{*}{\begin{tabular}{c}Control\\ Surgeons\\(n = 204)\end{tabular}} & \multirow{3}{*}{\begin{tabular}{c}Encouraged\\ Surgeons\\(n = 204)\end{tabular}} & \multirow{3}{*}{\begin{tabular}{c}Standardized\\ Difference\end{tabular}} & \multirow{3}{*}{\begin{tabular}{c}P-Value\end{tabular}} \\ \\ \\
\\
\multirow{2}{*}{\begin{tabular}[c]{@{}l@{}}Surgeons' preference: proportion\\ of CABG surgeries using TEE\end{tabular}} &  0.29   &0.84    &3.13    \\ 
 &       &     &   &      \\ 
 Cluater size &46.41 &45.49 &-0.08 &0.43 \\
\textsf{Composition of Patient Population} &           &          &          \\ 
\hspace{0.5 cm}Mean age, yrs & 75.35   & 75.46    & 0.09 & 0.36  \\
\hspace{0.5 cm}Percentage male, \% & 67.38     & 67.94    & 0.09 & 0.37      \\
\hspace{0.5 cm}Percentage white, \% & 91.35     & 91.73  & 0.06 & 0.52      \\ 
\hspace{0.5 cm}Percentage elective, \%  & 49.51   & 49.18    & -0.02 &0.81       \\ 
\hspace{0.5 cm}Percentage with diabetes, \%   & 13.57    & 13.34  & -0.04 & 0.69   \\
\hspace{0.5 cm}Percentage with renal diseases, \%   &7.25   & 7.29  & 0.01 &0.91   \\
\hspace{0.5 cm}Percentage with arrhythmia, \%   &10.41   & 9.99  & -0.08 &0.41   \\
\hspace{0.5 cm}Percentage with CHF, \%   & 8.92  & 8.88  & -0.01 & 0.93   \\
\hspace{0.5 cm}Percentage with hypertension, \%   & 25.38    & 25.02  & -0.04 & 0.66   \\
\hspace{0.5 cm}Percentage with obesity, \%   & 4.95    & 4.76  & -0.05 & 0.59   \\
\hspace{0.5 cm}Percentage with pulmonary diseases, \%   & 1.35  & 1.38  & 0.01 & 0.91   \\
\textsf{Surgeon/Hospital Characteristics} &           &          &          \\ 
\hspace{0.5 cm}Total cardiac surgical volume & 111.74 & 119.58 & 0.15 & 0.13 \\
\hspace{0.5 cm}Teaching hospital, yes/no  &0.09 &0.08 &-0.02 &0.51 \\
\hspace{0.5 cm}Total number of hospital beds &374.64 &387.00 &0.07 &0.47 \\
\hspace{0.5 cm}Full-time registered nurses  &597.50 &619.98 &0.06 &0.54 \\
\hspace{0.5 cm}Cardiac ICU, yes/no &0.84 & 0.82 & -0.07 &0.51 \\ 
\end{tabular}
\label{tbl: balance table}
\end{table}

\section{Notation and setup}
\label{sec: notation and assumption}
\subsection{General setup and IV assumptions}
\label{subsec: setup and IV assumptions}
Suppose there are $K$ matched pairs of two clusters, $k = 1, ..., K$, $j = 1, 2$, so that index $kj$ uniquely identifies one cluster. Each cluster $kj$ is associated with a continuous cluster-level instrumental variable $\widetilde{Z}_{kj}$, or an encouragement dose, and a vector of cluster-level covariates $\widetilde{\mathbf{x}}_{kj}$. In the application described in Section \ref{subsec: application in intro}, $K = 204$, $kj$ indexes surgeon $j$ in matched pair $k$, $\widetilde{Z}_{kj} \in [0, 1]$ measures surgeon $kj$'s preference for TEE, and $\widetilde{\mathbf{x}}_{kj}$ includes surgeon $kj$'s total cardiac surgical volume and characteristics of the hospital she worked at, e.g., total number of hospital beds. Each cluster $kj$ contains $n_{kj} \geq 1$ individuals, indexed by $i = 1,...,n_{kj}$, so that index $kji$ uniquely identifies individual $i$ in cluster $kj$. Each individual $kji$ is associated with a treatment indicator $D_{kji}$, outcome of interest $R_{kji}$, and individual-level covariates $\mathbf{x}_{kji}$. In our application, $n_{kj}$ is the number of patients treated by surgeon $kj$, $kji$ indexes the $i$th patient treated by surgeon $kj$, $D_{kji}$ is whether or not patient $kji$ receives TEE monitoring during her CABG surgery, $R_{kji}$ is patient $kji$'s 30-day mortality status, and $\mathbf{x}_{kji}$ includes patient $kji$'s age, gender, race, nature of the surgery (elective or not), and important comorbid conditions.

Write $\widetilde{\mathbf{z}} = (\widetilde{z}_{11}, ..., \widetilde{z}_{K2})$, $\widetilde{\mathbf{Z}} = (\widetilde{Z}_{11}, ..., \widetilde{Z}_{K2})$, $\mathbf{d} = (d_{111}, ..., d_{K2 n_{K2}})$, $\mathbf{D} = (D_{111}, ..., D_{K2 n_{K2}})$, $\mathbf{R} = (R_{111}, ..., R_{K2 n_{K2}})$. We consider the potential outcome framework to formalize an instrumental variable (\citealp{neyman1923application}; \citealp{rubin1974estimating}; \citealp{AIR1996}). Let $D_{kji}(\widetilde{\mathbf{Z}} = \widetilde{\mathbf{z}})$ denote whether individual $kji$ would receive treatment when the encouragement dose assignment is set to $\widetilde{\mathbf{z}}$, and $R_{kji}(\widetilde{\mathbf{Z}} = \widetilde{\mathbf{z}}, \mathbf{D} = \mathbf{d})$ the potential outcome individual $kji$ would exhibit under $\widetilde{\mathbf{Z}} = \widetilde{\mathbf{z}}$ and $\mathbf{D} = \mathbf{d}$. We consider the following identification assumptions:
\begin{enumerate}
    \item[A1] \textsf{Stable Unit Treatment Value Assumption (SUTVA)}: We assume that whether or not individual $kji$ would receive treatment depends only on the encouragement dose assignment of her own cluster, but not other clusters', so that $\widetilde{z}_{kj}=\widetilde{z}_{kj}^{\prime}$ implies $D_{kji}(\widetilde{\mathbf{z}}) = D_{kji}(\widetilde{\mathbf{z}}^{\prime})$, and $D_{kji}(\widetilde{\mathbf{Z}})$ can be simplified and written as $D_{kji}(\widetilde{Z}_{kj})$. Similarly, we assume that the potential outcome individual $kji$ would exhibit depends on $\widetilde{\mathbf{Z}}$ only through $\widetilde{Z}_{kj}$ and treatment assignments in her own cluster. Under this assumption, $R_{kji}(\widetilde{\mathbf{Z}}, \mathbf{D})$ can be simplified and written as $R_{kji}\big(\widetilde{Z}_{kj}, \mathbf{D}_{kj}(\widetilde{Z}_{kj})\big)$, where $ \mathbf{D}_{kj}(\widetilde{Z}_{kj})$ is a shorthand for $\big(D_{kj1}(\widetilde{Z}_{kj}), \cdots, D_{kjn_{kj}}(\widetilde{Z}_{kj})\big)$. Both assumptions are likely to hold in preference-based instrumental variable studies, because physicians' preference typically affects only patients she treated, and the health-related outcome exhibited by a patient is unlikely to depend on the treatment received of patients in other clusters. 
    \item[A2] \textsf{Exclusion Restriction}: We assume that the encouragement affects the outcome only via treatment assignment, so that $R_{kji}\big(\widetilde{Z}_{kj}, \mathbf{D}_{kj}(\widetilde{Z}_{kj})\big)$ can be further simplified to $R_{kji}\big(\mathbf{D}_{kj}(\widetilde{Z}_{kj})\big)$.
    
    \item[A3] \textsf{IV Relevance}: Let $\overline{D}_{kj} = n_{kj}^{-1}\cdot \sum_{i = 1}^{n_{kj}} D_{kji}$. We assume that for all $k$ and $j$, $\overline{D}_{kj}(\widetilde{z}_{kj}) \geq \overline{D}_{kj}(\widetilde{z}'_{kj})$ for all $\widetilde{z}_{kj} > \widetilde{z}'_{kj}$, and there exist $\widetilde{z}_{kj}, \widetilde{z}'_{kj}$ such that the strict inequality holds. In words, a higher encouragement dose does not decrease average cluster-level treatment received, and $\overline{D}_{kj}$ is not a constant function in $\widetilde{Z}_{kj}$.
    \item[A4] \textsf{IV Unconfoundedness}: We assume that the encouragement dose assignment probability is independent of potential outcomes conditional on relevant observed covariates. Specifically, we consider the following two conditions.
    
\begin{condition}[\rm A4-I]\rm
\label{condition: NUCA V1}
Suppose that $\widetilde{\mathbf{x}}_{kj}$ contains all relevant cluster-level covariates for cluster $kj$ and $\mathbf{x}_{kji}$ contains all relevant individual-level covariates for subject $kji$. Let $\mathbf{X}_{kj}$ be a $n_{kj} \times p$ matrix whose ith row is $\mathbf{x}_{kji}$, such that $\mathbf{X}_{kj}$ contains all individual-level information of cluster $kj$. We say that A4-I holds if
\begin{equation*}
\big\{R_{kji}(\mathbf{d}), \mathbf{D}_{kj}(\widetilde{z}), \mathbf{d} \in \mathcal{D}, \widetilde{z} \in \mathcal{\widetilde{Z}}, i = 1, \cdots, n_{kj}\big\} \indep \widetilde{Z}_{kj} \mid \widetilde{\mathbf{x}}_{kj}, \mathbf{X}_{kj}, ~\text{for all}~ k, j,
\end{equation*}
where $\mathcal{D}$ and $\mathcal{\widetilde{Z}}$ denote the set of values $\mathbf{d}$ and $\widetilde{z}$ can take, respectively.
\end{condition}

Assumption (A4-I) states that the assignment probability of the cluster-level instrumental variable depends on cluster-level covariates and all individual-level covariates within the cluster. In our application, this would be the case if a surgeon's preference for using TEE monitoring during CABG surgery depends on her annual surgical volume, hospital characteristics, and the entire distribution of her patients' age, gender, comorbidities, etc. In some circumstances, one may believe that the cluster-level encouragement dose assignment mechanism depends on individual-level covariates $\mathbf{X}_{kj}$ only via some cluster-level aggregate measures (\citealp{vanderweele2008ignorability}). This motivates the following relaxed condition.

\begin{condition}[\rm A4-II]\rm
\label{condition: NUCA V2}
 Let $h(\cdot): \mathbb{R}^{n_{kj} \times p} \mapsto \mathbb{R}^q$ be a vector-valued function that maps $\mathbf{X}_{kj}$ to a q-dimensional vector of aggregate measures that summarizes the covariate information of $n_{kj}$ individuals in cluster $kj$. We say that A4-II holds if 
\begin{equation*}
\big\{R_{kji}(\mathbf{d}), \mathbf{D}_{kj}(\widetilde{z}), \mathbf{d} \in \mathcal{D}, \widetilde{z} \in \mathcal{\widetilde{Z}}, i = 1, \cdots, n_{kj}\big\} \indep \widetilde{Z}_{kj} \mid \widetilde{\mathbf{x}}_{kj}, h(\mathbf{X}_{kj}), ~\text{for all}~ k, j.
\end{equation*}
\end{condition}
 Some commonly-used $h(\cdot)$ functions include the mean and quantile functions. 

\end{enumerate}


\subsection{Embedding observational instrumental variable data into a cluster-randomized encouragement experiment}
\label{subsec: embed observational data into exp}
Our approach to analyzing data after matching is to conduct randomization-based inference, where potential outcomes are held fixed and the only probability distribution that enters statistical inference is the law governing the encouragement assignment mechanism in each matched pair of two clusters. Recall that two clusters with similar observed covariates but distinct continuous encouragement doses are paired together. Let $\widetilde{\mathbf{Z}}_{\vee}=(\widetilde{Z}_{11}\vee \widetilde{Z}_{12}, \widetilde{Z}_{21}\vee \widetilde{Z}_{22}, \dots, \widetilde{Z}_{K1}\vee \widetilde{Z}_{K2})$ and $\widetilde{\mathbf{Z}}_{\wedge}=(\widetilde{Z}_{11}\wedge \widetilde{Z}_{12}, \widetilde{Z}_{21}\wedge \widetilde{Z}_{22}, \dots, \widetilde{Z}_{K1}\wedge \widetilde{Z}_{K2})$, where $a \vee b= \max(a,b)$ and $a \wedge b = \min(a,b)$, denote the maximum and minimum of two doses in each matched pair of two clusters. 

\begin{definition}
\label{def: potential outcomes}
Under assumptions (A1)-(A4), we define the following potential outcomes ($d_{Tkji}$, $d_{Ckji}$, $r_{Tkji}$, $r_{Ckji}$) associated with individual $kji$ after matching on relevant observed covariates:
\begin{equation*}
\begin{split}
    &d_{Tkji}\overset{\Delta}{=} D_{kji}(\widetilde{Z}_{kj}=\widetilde{Z}_{k1}\vee \widetilde{Z}_{k2}) = D_{kji}(\widetilde{Z}_{k1}\vee \widetilde{Z}_{k2}), \\ &d_{Ckji}\overset{\Delta}{=} D_{kji}(\widetilde{Z}_{kj}=\widetilde{Z}_{k1}\wedge \widetilde{Z}_{k2}) = D_{kji}(\widetilde{Z}_{k1}\wedge \widetilde{Z}_{k2}), \\
    &r_{Tkji}\overset{\Delta}{=} R_{kji}\big(\mathbf{D}_{kj}(\widetilde{Z}_{kj} = \widetilde{Z}_{k1}\vee \widetilde{Z}_{k2})\big) = R_{kji}\big(\mathbf{D}_{kj}(\widetilde{Z}_{k1}\vee \widetilde{Z}_{k2})\big), \\
    &r_{Ckji}\overset{\Delta}{=} R_{kji}\big(\mathbf{D}_{kj}(\widetilde{Z}_{kj} = \widetilde{Z}_{k1}\wedge \widetilde{Z}_{k2})\big) = R_{kji}\big(\mathbf{D}_{kj}(\widetilde{Z}_{k1}\wedge \widetilde{Z}_{k2})\big).
\end{split}
\end{equation*}
\end{definition}
\begin{remark}\rm \label{remark: inteference}
The counterfactual outcome $r_{Tkji}$ (or $r_{Ckji}$) describes individual $kji$'s potential outcome when all individuals in that cluster receive encouragement dose $\widetilde{Z}_{kj}=\widetilde{Z}_{k1}\vee \widetilde{Z}_{k2}$ (or $\widetilde{Z}_{kj}=\widetilde{Z}_{k1}\wedge \widetilde{Z}_{k2}$); this definition allows interference between individuals in each cluster, and $\mathbf{D}_{kj}(\widetilde{Z}_{kj} = \widetilde{Z}_{k1}\vee \widetilde{Z}_{k2})$ describes only one realization of the length-$n_{kj}$ vector of treatment indicators $\mathbf{D}_{kj}$. $\mathbf{D}_{kj}$ has other realizations, and $R_{kji}$ may change under other $\mathbf{D}_{kj}$ configurations. However, the available data only speaks to the issue when the entire cluster receives encouragement of a certain dose and $\mathbf{D}_{kj}$ is set to its natural level under this encouragement dose.
\end{remark}

Write $\mathcal{F} = \big\{(\widetilde{\mathbf{x}}_{kj}, \mathbf{x}_{kji}, d_{Tkji}, d_{Ckji}, r_{Tkji}, r_{Ckji}): k = 1, \dots, K, ~j = 1,2, ~i = 1,\dots,n_{kj}\big\}$. The law that describes the encouragement dose assignment in each matched pair of two clusters is:
\[
\pi_{k1} = \text{pr}(\widetilde{Z}_{k1} = \widetilde{Z}_{k1} \wedge \widetilde{Z}_{k2}, \widetilde{Z}_{k2} = \widetilde{Z}_{k1}\vee \widetilde{Z}_{k2}\mid \mathcal{F}, \widetilde{\mathbf{Z}}_{\vee}, \widetilde{\mathbf{Z}}_{\wedge}),
\] and $\pi_{k2} = 1 - \pi_{k1}$. When the encouragement is indeed randomized as in a cluster-randomized trial, we necessarily have $\pi_{k1} = \pi_{k2} = 1/2$. In an observational study, however, we can only hope to make $\pi_{k1} \approx \pi_{k2}$ via matching on relevant observed covariates or the estimated propensity score (\citealp{rosenbaum1983central}). Write conditional density function
\begin{align*}
  &f\left(\widetilde{Z}_{kj} = \widetilde{Z}_{k1} \wedge \widetilde{Z}_{k2} \mid \widetilde{\mathbf{x}}_{kj}, h(\mathbf{X}_{kj})\right)= \xi\left\{\widetilde{Z}_{k1} \wedge \widetilde{Z}_{k2}, \widetilde{\mathbf{x}}_{kj}, h(\mathbf{X}_{kj})\right\} \\
    \text{and}~~&f\left(\widetilde{Z}_{kj} = \widetilde{Z}_{k1} \vee \widetilde{Z}_{k2} \mid \widetilde{\mathbf{x}}_{kj}, h(\mathbf{X}_{kj})\right) = \xi\left\{\widetilde{Z}_{k1} \vee \widetilde{Z}_{k2}, \widetilde{\mathbf{x}}_{kj}, h(\mathbf{X}_{kj})\right\},
\end{align*}
for $j = 1,2$. Under (A4-II), we have
\begin{equation*}\small
    \begin{split}
        &\pi_{k1} = \text{pr}(\widetilde{Z}_{k1} = \widetilde{Z}_{k1} \wedge \widetilde{Z}_{k2}, \widetilde{Z}_{k2} = \widetilde{Z}_{k1}\vee \widetilde{Z}_{k2}\mid \mathcal{F}, \widetilde{\mathbf{Z}}_{\vee}, \widetilde{\mathbf{Z}}_{\wedge})\\
        &= \frac{\xi\{\widetilde{Z}_{k1} \wedge \widetilde{Z}_{k2}, \widetilde{\mathbf{x}}_{k1}, h(\mathbf{X}_{k1})\} \cdot \xi\{\widetilde{Z}_{k1} \vee \widetilde{Z}_{k2}, \widetilde{\mathbf{x}}_{k2}, h(\mathbf{X}_{k2})\}}{\xi\{\widetilde{Z}_{k1} \wedge \widetilde{Z}_{k2}, \widetilde{\mathbf{x}}_{k1}, h(\mathbf{X}_{k1})\}\cdot \xi\{\widetilde{Z}_{k1} \vee \widetilde{Z}_{k2}, \widetilde{\mathbf{x}}_{k2}, h(\mathbf{X}_{k2})\} + 
        \xi\{\widetilde{Z}_{k1} \wedge \widetilde{Z}_{k2}, \widetilde{\mathbf{x}}_{k2}, h(\mathbf{X}_{k2})\} \cdot \xi\{\widetilde{Z}_{k1} \vee \widetilde{Z}_{k2}, \widetilde{\mathbf{x}}_{k1}, h(\mathbf{X}_{k1})\}},
    \end{split}
\end{equation*}
and one sufficient condition for $\pi_{k1} = \pi_{k2} = 1/2$ is $\{\widetilde{\mathbf{x}}_{k1}, h(\mathbf{X}_{k1})\} = \{\widetilde{\mathbf{x}}_{k2}, h(\mathbf{X}_{k2})\}$. If one believes that (A4-I) is more likely to hold, one needs to further make $\mathbf{X}_{k1} \approx \mathbf{X}_{k2}$ in the design stage, in addition to their dimension reduction. One way to enforce $\mathbf{X}_{k1} \approx \mathbf{X}_{k2}$ is to further match similar individuals within matched pair of two clusters. We refer readers to \citet{zubizarreta2017optimal} and \citet{pimentel2018optimal} for more details on statistical matching methods for multilevel data. 

\subsection{When is a cluster-level design and analysis preferred over an individual-level one?}
\label{subsec: why is cluster better}
Instead of pairing comparable clusters and embedding data into a cluster-randomized encouragement experiment, empirical researchers sometimes disregard the clustering structure and directly pair individuals across different clusters, e.g., pairing one patient from a high-TEE-preference surgeon to one from a low-TEE-preference surgeon. When would researchers prefer a cluster-level design (e.g., pairing surgeons) to a non-clustered, individual-level, design (e.g., pairing patients)? Below we discuss some important motivations and practical advantages of a cluster-level design and analysis.

\begin{enumerate}
    \item First, when a cluster-randomized experiment is preferred over a non-clustered one when researchers formulate the causal question in terms of a hypothetical randomized experiment, a conceptual stage preceding the design stage (\citealp{bind2019bridging}). A cluster-randomized (encouragement) experiment would be preferred when the encouragement (or treatment in non-IV settings) is intrinsically cluster-level and impractical to be applied at the individual level, e.g., a public health campaign that targets the entire communities. For our application, individual patient does not get to choose whether her surgeon uses TEE monitoring or not during her CABG surgery; surgeons choose to use TEE to monitor and manage their patients' hemodynamics during the surgery. Therefore, it is only meaningful to encourage surgeons to use TEE monitoring in a hypothetical randomized experiment, and our design is meant to replicate this hypothetical cluster-randomized encouragement experiment.
    
    \item Second, when unmeasured IV-outcome confounding is still a concern after adjusting for observed confounding variables. In Supplementary Material E, we demonstrated from a causal directed acyclic graph (DAG) perspective and via extensive simulations that a cluster-level primary analysis is less biased compared to an individual-level primary analysis when individual-level unmeasured confounders and observed confounders are correlated via a shared cluster-level latent factor.  
    
    \item Third, \citet{hansen2014clustered} found that for a binary instrumental variable in a favorable situation where there is a genuine treatment effect and no unmeasured confounding, clustered treatment assignment exhibits larger insensitivity to hidden bias when researchers conduct a sensitivity analysis. In other words, clustered treatment assignment exhibits larger design sensitivity (\citealp{rosenbaum2004design}), i.e., asymptotically a larger power in a sensitivity analysis.
    
    \item Fourth, when the individual-level no interference assumption is inappropriate (\citealp{rosenbaum2007interference}; \citealp{hudgens2008toward}). The definition of potential outcomes in a cluster-randomized experiment does not require assuming no interference among individuals within each cluster (see Remark~\ref{remark: inteference}). However, no interference is a necessary assumption when researchers design and conduct randomization inference at the individual-level (\citealp{small2008war}; \citealp{baiocchi2010building}).
\end{enumerate}

\section{Randomization-based inference}
\label{sec: random-based inference}
\subsection{Cluster-level sharp null hypothesis}
\label{subsec: sharp null and constant prop model}
Consider testing the following cluster-level sharp null hypothesis:
\begin{equation}\small
\label{hypothesis: general sharp null}
\begin{split}
    H_{0, \text{sharp}}: ~~&f_{kj}(r_{Tkj1}, \dots, r_{Tkjn_{kj}}, d_{Tkj1}, \dots, d_{Tkjn_{kj}})\\
    = ~&f_{kj}(r_{Ckj1}, \dots, r_{Ckjn_{kj}}, d_{Ckj1}, \dots, d_{Ckjn_{kj}}), ~\text{for all}~ k, j,
\end{split}
\end{equation}
where $f_{kj}(\cdot): \mathbb{R}^{2n_{kj}} \rightarrow \mathbb{R}$ maps the $2n_{kj}$ potential outcomes $\big\{(d_{Tkji}, r_{Tkji}), ~i = 1, \cdots, n_{kj}\big\}$ or $\big\{(d_{Ckji}, r_{Ckji}), ~i = 1, \cdots, n_{kj}\big\}$ to a scalar aggregate outcome.

\begin{example}
\label{ex: prop on cluster outcome}
When $f_{kj}(x_{1},\cdots, x_{n_{kj}}, y_{1},\dots, y_{n_{kj}})=n_{kj}^{-1}\sum_{i=1}^{n_{kj}}x_{i}-\beta \cdot n_{kj}^{-1}\sum_{i=1}^{n_{kj}}y_{i}$, \eqref{hypothesis: general sharp null} reduces to testing $H_{0, \text{sharp}, \text{prop}}: \beta = \beta_0$ in the following constant proportional treatment effect model:
\begin{equation}
    \label{eqn: constant prop model}
    n_{kj}^{-1}\sum_{i=1}^{n_{kj}}r_{Tkji}-n_{kj}^{-1}\sum_{i=1}^{n_{kj}}r_{Ckji}=\beta \left(n_{kj}^{-1}\sum_{i=1}^{n_{kj}}d_{Tkji}-n_{kj}^{-1}\sum_{i=1}^{n_{kj}}d_{Ckji}\right), ~\text{for all}~ k, j,
\end{equation}
which states that the mean difference of individuals' potential outcomes under encouragement and control is proportional to the mean difference of individuals' potential treatment received under encouragement and control. When $n_{kj} = 1$, $H_{0, \text{sharp}, \text{prop}}$ reduces to the proportional treatment effect model considered in \citet{small2008war}. Setting $\beta=0$ yields a cluster-level sharp null hypothesis of no treatment effect: $H_{0, \text{cluster}}: n_{kj}^{-1}\sum_{i=1}^{n_{kj}}r_{Tkji}=n_{kj}^{-1}\sum_{i=1}^{n_{kj}}r_{Ckji}$ for all $k, j$.
\end{example}

\begin{example}
\label{ex: each unit's sharp null}
Consider testing an individual-level sharp null hypothesis in a clustered design, i.e., $H_{0, \text{unit}}: r_{Tkji} = r_{Ckji}$ for all $k$, $j$, and $i$ as in \citet{small2008randomization}. It is easy to see that $H_{0, \text{unit}}$ implies $H_{0, \text{cluster}}$, and any statistic testing $H_{0, \text{cluster}}$ is also a valid test statistic for $H_{0, \text{unit}}$. Similarly, consider testing $H_{0, \text{unit}, \text{prop}}: \beta = \beta_0$ in a individual-level constant proportional treatment effect model: $r_{Tkji} - r_{Ckji} = \beta (d_{Tkji} - d_{Ckji})$ for all $k$, $j$, and $i$ as in \citet{small2008war}. Again, it is easy to see that $H_{0, \text{unit}, \text{prop}}$ implies the cluster-level null hypothesis $H_{0, \text{sharp}, \text{prop}}$.
\end{example}

In Supplementary Material C, we derived a rich class of nonparametric test statistics to test the null hypothesis $H_{0, 
\text{sharp}}$. The proposed family of test statistics $T_{\text{DR}}$ incorporates encouragement dose information, and contains many familiar test statistics as special cases, including the sign test, the Wilcoxon signed rank test, the dose-weighted signed rank test (\citealp{rosenbaum1997signed}), and the polynomial rank test (\citealp{rosenbaum2007confidence}).

\subsection{Pooled effect ratio (PER)}
\label{subsec: effect ratio}
Fisher's sharp null hypothesis may be restricted in cluster-randomized designs as heterogeneity may be expected across different clusters. In this section, we extend an estimand allowing for treatment heterogeneity known as the effect ratio (\citealp{imai2009essential}; \citealp{baiocchi2010building,kang2018estimation}) to paired cluster-randomized encouragement experiments, and derive randomization-based inferential methods.

\begin{definition}
\label{def: pooled effect ration}
Pooled effect ratio (PER), $\lambda_{\text{PER}}$, refers to the following quantity:
\begin{equation*}
    \lambda_{\text{PER}} = \frac{\sum_{k = 1}^K \sum_{j = 1}^2 \sum_{i = 1}^{n_{kj}} (r_{Tkji} - r_{Ckji})}{\sum_{k = 1}^K \sum_{j = 1}^2 \sum_{i = 1}^{n_{kj}} (d_{Tkji} - d_{Ckji})},
\end{equation*}
where counterfactuals $(r_{Tkji}, r_{Ckji}, d_{Tkji}, d_{Ckji})$ are defined in Definition \ref{def: potential outcomes}.
\end{definition}

Definition \ref{def: pooled effect ration} assumes that $\sum_{k = 1}^K \sum_{j = 1}^2 \sum_{i = 1}^{n_{kj}} (d_{Tkji} - d_{Ckji}) \neq 0$. Consider testing the null hypothesis $H_{0, \text{PER}}:\lambda_{\text{PER}}=\lambda_0$ using the statistic $T(\lambda_0)=K^{-1}\sum_{k=1}^K Y_k(\lambda_0)$, where $Y_k(\lambda_0)=\sum_{j=1}^2 (2Z_{kj}-1) \left(\sum_{i=1}^{n_{kj}}R_{kji}-\lambda_0 \sum_{i=1}^{n_{kj}}D_{kji}\right)$. Proposition \ref{prop: CLT effect ratio} characterizes useful properties of the statistic $T(\lambda_0)$.

\begin{proposition}
\label{prop: CLT effect ratio}
Under $H_{0, \text{PER}}: \lambda_{\text{PER}}=\lambda_0$ and $\mathbb{E}\big\{Z_{kj}\mid \mathcal{F}, \widetilde{\mathbf{Z}}_{\vee}, \widetilde{\mathbf{Z}}_{\wedge}\big\} = 1/2$ for all $k$ and $j$, it is true that $\mathbb{E}\big\{T(\lambda_0)\mid\mathcal{F}, \widetilde{\mathbf{Z}}_{\vee}, \widetilde{\mathbf{Z}}_{\wedge}\big\}=0$, and
${\rm var} \big\{\sqrt{K}T(\lambda_0)\mid\mathcal{F}, \widetilde{\mathbf{Z}}_{\vee}, \widetilde{\mathbf{Z}}_{\wedge}\big\}=K^{-1}\sum_{k=1}^K {\rm var}\big\{Y_k(\lambda_0)\mid\mathcal{F}, \widetilde{\mathbf{Z}}_{\vee}, \widetilde{\mathbf{Z}}_{\wedge} \big\}$, where
 ${\rm var} \big\{Y_k(\lambda_0)\mid\mathcal{F}, \widetilde{\mathbf{Z}}_{\vee}, \widetilde{\mathbf{Z}}_{\wedge}\big\} =
 4^{-1} \big\{ \sum_{i=1}^{n_{k1}} (r_{Tk1i}-\lambda_0 d_{Tk1i}+r_{Ck1i}-\lambda_0 d_{Ck1i})-\sum_{i=1}^{n_{k2}} (r_{Tk2i}-\lambda_0 d_{Tk2i}+r_{Ck2i}-\lambda_0 d_{Ck2i})\big\}^2$. Under conditions S1-S2 in the Supplementary Material B, as $K\rightarrow\infty$ and conditional on $\mathcal{F}, \widetilde{\mathbf{Z}}_{\vee}, \widetilde{\mathbf{Z}}_{\wedge}$, 
\begin{equation}
   \frac{ \sqrt{K} T(\lambda_0) }{\sqrt{ K^{-1}\sum_{k=1}^K {\rm var}\big\{Y_k(\lambda_0)\mid\mathcal{F}, \widetilde{\mathbf{Z}}_{\vee}, \widetilde{\mathbf{Z}}_{\wedge} \big\}}} ~\text{converges in distribution to}~ \text{Normal}(0,1). \label{eqn: CLT effect ratio}
\end{equation}
\end{proposition}

\begin{proof}
All proofs in this article can be found in the Supplementary Material B.
\end{proof}

Observe that ${\rm var}\big\{Y_k(\lambda_0)\mid\mathcal{F}, \widetilde{\mathbf{Z}}_{\vee}, \widetilde{\mathbf{Z}}_{\wedge} \big\}$ in Proposition \ref{prop: CLT effect ratio} depends on both potential outcomes, one of which is always missing. To construct a valid test statistic, we need to derive a variance estimator based on observed data. Classical literature in paired experiments typically adopt the following sample variance of the observed paired differences as a conservative variance estimator \citep{Imai2008variancepaird, baiocchi2010building}:
\begin{equation*}
   S^2(\lambda_0) = \frac{1}{K-1} \sum_{k=1}^K \big\{Y_k(\lambda_0)- T(\lambda_0)\big\}^2. 
\end{equation*}
Recent works by \citet{lin2013agnostic}, \cite{ding2016jrssb} and \cite{fogarty2018finelystratified, fogarty2018paired} developed regression-assisted variance estimators that remain conservative in expectation, but can be less conservative than using the sample variance. To adapt this idea to our current setting, let $Q$ be an arbitrary $K\times p$ matrix with $p<K$, and $H_Q=Q(Q^T Q)^{-1}Q^T$ its hat matrix. Let $h_{Qk}$ be the $k$th diagonal element of $H_Q$, and $Y_Q$ a $K\times 1$ column vector with $k^{\text{th}}$ entry $Y_k(\lambda_0)/\sqrt{ 1- h_{Qk}}$. Define the following variance estimator:
\begin{equation}
    S^2_Q(\lambda_0) = \frac{1}{K} Y_Q^T (I-H_Q) Y_Q. 
    \label{eqn: S^2_Q}
\end{equation}
We proved in the Supplementary Material B that $S_Q^2(\lambda_0)$ is always a conservative variance estimator for $\text{var}\{\sqrt{K}T(\lambda_0)\}$ in finite sample. When $Q=\bm{e}$, a column vector of 1's, $S^2_Q(\lambda_0)$ is precisely equal to the classical variance estimator $S^2(\lambda_0)$.
More generally, $Q$ may contain any cluster-level or even individual-level covariate information. Researchers may consider two types of $Q$ matrix. The first type, denoted by $Q_1=[\bm{e}, \bm{B}]$, adjusts for cluster-level covariate information. Each row of $\bm{B}$ contains centered cluster-level covariates  $\widetilde{\bm{x}}_{k1}$ and $\widetilde{\bm{x}}_{k2}$. Centering the covariates renders $\bm{e}$ and $\bm{B}$ orthogonal, which is not necessary but makes theoretical derivations easier. The second type, denoted by $Q_2=[Q_1, \bm{W}]$, further contains individual-level covariate information, e.g., averages of individual covariates within each cluster $\overline{\bm{X}}_{k1}$ and $\overline{\bm{X}}_{k2}$, and $\bm{W}$ is the residual after projecting $\overline{\bm{X}}_{k1}$ and $\overline{\bm{X}}_{k2}$ onto the column space of $Q_1$. 

Proposition \ref{prop: variance improvement} suggests that the variance estimator $S^2_Q(\lambda_0)$ becomes less conservative as we use finer covariate information that are predictive of the treatment effect heterogeneity across $K$ cluster pairs, which eventually leads to more powerful inference.

\begin{proposition} \label{prop: variance improvement}
Under the assumptions in Proposition 1 and condition S3 in the Supplementary Material B, as $K\rightarrow \infty$,
\begin{equation}
   S^2_{Q} (\lambda_0) - {\rm var} \big\{\sqrt{K}T(\lambda_0)\mid\mathcal{F}, \widetilde{\mathbf{Z}}_{\vee}, \widetilde{\mathbf{Z}}_{\wedge}\big\}  ~\text{converges in probability to}~ \lim_{K\rightarrow\infty} \frac{1}{K}\bm{\tau}^T (I-H_Q)\bm{\tau},
\end{equation}
where $\bm{\tau}$ is a length-K vector with each entry being $\mathbb{E}\big\{Y_k(\lambda_0)\mid\mathcal{F}, \widetilde{\mathbf{Z}}_{\vee}, \widetilde{\mathbf{Z}}_{\wedge}\big\}$. If we further assume that $Q\in\{\bm e, Q_1, Q_2\}$ satisfies condition S3, as $K\rightarrow \infty$,
\begin{equation*}
    \begin{split}
        & S^2_{\bm{e}}(\lambda_0) - S^2_{Q_1}(\lambda_0) ~\text{converges in probability to}~ {\bm \beta}_B^T \Sigma^{-1}_B{\bm \beta}_B \geq 0,\\
    \text{and}\quad&S^2_{Q_1}(\lambda_0) - S^2_{Q_2}(\lambda_0) ~\text{converges in probability to}~ {\bm \beta}_W^T \Sigma^{-1}_W{\bm \beta}_W \geq 0,      
    \end{split}
\end{equation*}
where $\Sigma_B=\lim _{K\rightarrow\infty} K^{-1} \bm{B}^T \bm{B}$, $\Sigma_W=\lim _{K\rightarrow\infty} K^{-1} \bm{W}^T \bm{W}$, $\bm{\beta}_B=  \lim_{K\rightarrow\infty} K^{-1}  {\bm B}^T \bm{\tau}$, and $\bm{\beta}_W=  \lim_{K\rightarrow\infty} K^{-1}  {\bm W}^T \bm{\tau}$.
\end{proposition}

Proposition \ref{prop: variance improvement} does not rely on correctly specifying the relationship between $Y_k(\lambda_0)$ and covariates. Under mild regularity conditions, it is shown in the Supplementary Material B that $S_Q^2(\lambda_0)$ is approximately the mean-squared error after regressing $Y_k(\lambda_0)$ on the relevant covariates. In this sense, utilizing $Q$ in general produces a less conservative variance estimator, similar to the non-clustered settings discussed in \citet{lin2013agnostic}, \citet{ding2016jrssb}, and \citet{fogarty2018finelystratified,fogarty2018paired}. The magnitude of improvement depends on how well $Q$ predicts $Y_k(\lambda_0)$. Given the variance estimator $S^2_Q$, for large $K$, the hypothesis $H_{0, \text{PER}}: \lambda_{\text{PER}}=\lambda_0$ can be tested by comparing the statistic (\ref{eqn: CLT effect ratio}) with $S^2_Q$ in place of $K^{-1}\sum_{k=1}^K {\rm var}\big\{Y_k(\lambda_0)\mid\mathcal{F}, \widetilde{\mathbf{Z}}_{\vee}, \widetilde{\mathbf{Z}}_{\wedge}\big\}$ to the standard normal distribution.

\subsection{Simpson's paradox}
\label{subsec: from effect ratio to ACTE}
The pooled effect ratio is an appealing measure but does have one concerning feature. Consider two studies of the same scientific interest but conducted at two years. For simplicity, suppose that we have a binary instrumental variable and two clusters indexed by $k = 1, 2$, each with $n_k$ individuals. Let $a_k = \sum_{i = 1}^{n_{k}} (r_{Ti} - r_{Ci})$ and $b_k = \sum_{i = 1}^{n_{k}} (d_{Ti} - d_{Ci})$ denote the causal effect of the instrument on the treatment received and outcome in each cluster in the first year, and $a^\prime_k$ and $b^\prime_k$ the second year. Suppose that 
\[
\frac{a_{k}}{b_{k}} < \frac{a^\prime_k}{b^\prime_k},~ k \in \{1, 2\},
\]
i.e., both clusters exhibit a larger cluster-specific effect ratio in year two compared to year one. However, the pooled effect ratio may fail to preserve this trend, such that the following holds:
\[
\text{pooled effect ratio in year one} = \frac{a_1 + a_2}{b_1 + b_2} >\frac{a^\prime_1 + a^\prime_2}{b^\prime_1 + b^\prime_2} = \text{pooled effect ratio in year two}.
\]
This phenomenon is a particular instance of the well-known Simpson's paradox (\citealp{blyth1972simpson}; \citealp{bickel1975sex}): a trend appears in each subgroup or stratum, but disappears when subgroups or strata are combined. This drawback of pooled effect ratio motivates us to consider a new estimand that is trend-preserving while still allows for treatment heterogeneity.

\section{A cluster-heterogeneous proportional treatment effect model and average cluster effect ratio (ACER)} \label{sec: ACER}
\subsection{A cluster-heterogeneous structural model and a new estimand}
Consider the following cluster-heterogeneous proportional treatment effect model:
\begin{equation}
    \label{eqn: cluster heterogeneous prop model}
    \sum_{i = 1}^{n_{kj}} r_{Tkji} - \sum_{i = 1}^{n_{kj}} r_{Ckji} = \beta_{kj}\left(\sum_{i = 1}^{n_{kj}}d_{Tkji} - \sum_{i = 1}^{n_{kj}}d_{Ckji}\right), ~\text{for all}~ k = 1, \cdots, K,~j = 1, 2,
\end{equation}
where $\beta_{kj}$ measures a cluster-specific proportional treatment effect, and can be heterogeneous across different clusters. Instead of testing $H_{0, \text{sharp}, \text{prop}}: \beta_{kj} = \beta$ for all $k$ and $j$ as in Example \ref{ex: prop on cluster outcome}, we define 
\begin{equation*}
    \frac{1}{2K} \sum_{k = 1}^K\sum_{j = 1}^2 \beta_{kj} = \overline{\beta}_{kj} = \lambda_{\text{ACER}},
\end{equation*}
and consider testing the following weak null hypothesis:
\begin{equation}
\label{eqn: weak null hypothesis}
    H_{0, \text{ACER}}: \lambda_{\text{ACER}} = \lambda_0.
\end{equation}
The new estimand $\lambda_{\text{ACER}}$ is referred to as the average cluster effect ratio (ACER). 



\begin{remark}\rm
The average cluster effect ratio (ACER) is different from the pooled effect ratio (PER): \[
\frac{1}{2K} \sum_{k = 1}^K\sum_{j = 1}^2 \beta_{kj} =
\underbrace{\frac{1}{2K}\sum_{k = 1}^K\sum_{j = 1}^2 \frac{\sum_{i = 1}^{n_{kj}} r_{Tkji} - \sum_{i = 1}^{n_{kj}} r_{Ckji}}{\sum_{i = 1}^{n_{kj}}d_{Tkji} - \sum_{i = 1}^{n_{kj}}d_{Ckji}}}_{\text{ACER}} \neq \underbrace{\frac{\sum_{k = 1}^K \sum_{j = 1}^2 \sum_{i = 1}^{n_{kj}} (r_{Tkji} - r_{Ckji})}{\sum_{k = 1}^K \sum_{j = 1}^2 \sum_{i = 1}^{n_{kj}} (d_{Tkji} - d_{Ckji})}}_{\text{PER}}.
\]
It is easy to see that the average cluster effect ratio no longer suffers from Simpson's paradox, as it is now monotonic in each $\beta_{kj}$.
\end{remark}

\begin{remark}\rm
In our application, $\beta_{kj}$ measures a surgeon-specific effect of TEE monitoring on clinical outcomes. If the effect of TEE monitoring is believed to vary with practicing surgeons' expertise and experience, a cluster-heterogeneous proportion treatment effect model may be more appropriate than a constant proportional treatment effect model.
\end{remark}

\subsection{Developing an asymptotically valid test for the average cluster effect ratio}
\label{subsec: develop the weak null test}
Define $\text{CO}_{kj} = \sum_{i = 1}^{n_{kj}}d_{Tkji} - \sum_{i = 1}^{n_{kj}}d_{Ckji}$ and let $\iota_{kj} = \text{CO}_{kj}/n_{kj}$. Note that $\iota_{kj}$ is always well-defined and can be interpreted as cluster $kj$'s compliance rate if we further assume no defiers. The IV relevance assumption (A3) implies that:
\[
\iota_{kj} = n^{-1}_{kj} \cdot \left(\sum_{i = 1}^{n_{kj}}d_{Tkji} - \sum_{i = 1}^{n_{kj}}d_{Ckji}\right) = \overline{D}_{kj}(\widetilde{Z}_{k1}\vee \widetilde{Z}_{k2}) - \overline{D}_{kj}(\widetilde{Z}_{k1}\wedge \widetilde{Z}_{k2})\geq 0.
\]
Observe that according to model (\ref{eqn: cluster heterogeneous prop model}), the structural parameter $\beta_{kj}$ is not well-defined in cluster $kj$ after matching unless the following stronger IV relevance assumption (A3') holds: 
\begin{enumerate}\rm
    \item[A3'] \textsf{Strict IV Relevance After Matching}: We assume that $\iota_{kj}$ is bounded away from $0$ for each cluster $kj$ after matching:
    \begin{equation}
           \inf_{k, j} \iota_{kj} = \iota_{\text{min}} > 0.
           \label{eqn: lower bound on compliance rate}
    \end{equation}
\end{enumerate}
In words, assumption (A3') says that for the $k^{\text{th}}$ matched pair of two clusters, the encouragement dose $\widetilde{Z}_{k1}\vee \widetilde{Z}_{k2}$ compared to $\widetilde{Z}_{k1}\wedge \widetilde{Z}_{k2}$ would change the cluster-aggregated treatment received in cluster $kj$ by at least $n_{kj}\cdot\iota_{\text{min}}$ units. Requiring $\iota_{kj} \geq \iota_{\text{min}} > 0$ is a minimal assumption to have $\beta_{kj}$ and therefore the estimand $\lambda_{\text{ACER}}$ well-defined.

\begin{remark}\rm
Although $\iota_{\text{min}} > 0$ suffices for identification, it allows the encouragement defined by $\widetilde{Z}_{k1}\vee \widetilde{Z}_{k2}$ compared to $\widetilde{Z}_{k1}\wedge \widetilde{Z}_{k2}$ to be an arbitrarily weak instrument, i.e., $\iota_{kj}$ can be arbitrarily close to $0$. Confidence intervals obtained from weak instruments are well-known to be long and non-informative (\citealp{imbens2005robust}). To remedy this, some researchers have advocated strengthening an instrument in the design stage of an observational study (\citealp{baiocchi2010building, keele2016strong}) by penalizing two clusters having close encouragement doses $\widetilde{Z}_{kj}$ during statistical matching and forcing two units in each matched pair to be markedly different in $\widetilde{Z}_{kj}$. It may often be reasonable to further assume that, after strengthening, the instrument is not exceptionally weak and set $\iota_{\text{min}}$ to some conservative (e.g., $\iota_{\text{min}} = 0.10$) albeit non-degenerate value (e.g., $\iota_{\text{min}} = 0.0001$). The test developed in this section is valid under the minimal assumption $\iota_{\text{min}} > 0$; however, stronger, design-driven, assumptions on $\iota_{\text{min}}$ could help largely shorten the confidence interval and increase efficiency.
\end{remark}

\begin{remark}\rm
\label{remark: PER ACER compliance}
Recall that 
\begin{equation*}
\begin{split}
    \text{Pooled Effect Ratio} &= \frac{\sum_{k = 1}^K \sum_{j = 1}^2 \sum_{i = 1}^{n_{kj}} (r_{Tkji} - r_{Ckji})}{\sum_{k = 1}^K \sum_{j = 1}^2 \sum_{i = 1}^{n_{kj}} (d_{Tkji} - d_{Ckji})} \\
    &= \sum_{k = 1}^K\sum_{j = 1}^2 \underbrace{\frac{\sum_{i = 1}^{n_{kj}}(d_{Tkji} - d_{Ckji})}{\sum_{k = 1}^K \sum_{j = 1}^2 \sum_{i = 1}^{n_{kj}} (d_{Tkji} - d_{Ckji})}}_{w_{kj}}\cdot \underbrace{\frac{\sum_{i = 1}^{n_{kj}}(r_{Tkji} - r_{Ckji})}{\sum_{i = 1}^{n_{kj}}(d_{Tkji} - d_{Ckji})}}_{\beta_{kj}}.
\end{split}
\end{equation*}
The pooled effect ratio (PER) can be viewed as a weighted average of each cluster's effect ratio $\beta_{kj}$, where $\beta_{kj}$ is weighted by the proportion of cluster $kj$'s compliers among all compliers (assuming no defiers). Therefore, if the encouragement dose $\widetilde{Z}_{k1}\vee \widetilde{Z}_{k2}$ compared to $\widetilde{Z}_{k1}\wedge \widetilde{Z}_{k2}$ is an exceptionally weak instrument for cluster $kj$ so that the weight $w_{kj} \approx 0$, $\beta_{kj}$ does not contribute much to the final pooled effect ratio estimand. In this way, the pooled effect ratio estimand gives most weight to large, high-compliance-rate clusters, does not require the stricter IV relevance assumption (A3'), and skirts the potential weak instrument problem that some or even many clusters may have low compliance rate.  
\end{remark}

Consider the encouraged cluster $j$ in pair $k$ with $Z_{kj} = 1$. $\text{CO}_{kj}$ counts the number of compliers (those with $(d_T, d_C) = (1, 0)$) minus the number of defiers (those with $(d_T, d_C) = (0, 1)$) in cluster $kj$. The number of defiers in cluster $kj$ is lower bounded by $0$, and the number of compliers is upper bounded by $Z_{kj}\sum_{i=1}^{n_{kj}}D_{kjn_{kj}}$; hence, $\text{CO}_{kj}$ is upper bounded by $Z_{kj}\sum_{i=1}^{n_{kj}}D_{kjn_{kj}}$, an observed quantity. Similarly, $\text{CO}_{kj}$ is upper bounded by $(1-Z_{kj})\sum_{i=1}^{n_{kj}}(1-D_{kjn_{kj}})$ in each control cluster. Combining this with \eqref{eqn: lower bound on compliance rate}, we have the following box constraint on $\text{CO}_{kj}$:
\begin{equation}
\label{eqn: box constraint}
\begin{split}
&n_{kj}\cdot \iota_{\text{min}} \leq \text{CO}_{kj} \leq Z_{kj}\sum_{i=1}^{n_{kj}}D_{kjn_{kj}}+(1-Z_{kj})\sum_{i=1}^{n_{kj}}(1-D_{kjn_{kj}}), ~\text{for all}~k, j.
\end{split}
\end{equation}

\begin{remark}\rm
The box constraint \eqref{eqn: box constraint} suggests that assumption (A3') may be invalidated by the observed data: we know that (A3') fails to hold for cluster $kj$ when
\[
Z_{kj}\sum_{i=1}^{n_{kj}}D_{kjn_{kj}}+(1-Z_{kj})\sum_{i=1}^{n_{kj}}(1-D_{kjn_{kj}}) < n_{kj}\cdot \iota_{\text{min}}.
\]
This is the case, for instance, when there is no individual receiving treatment in an encouraged cluster, or no individual not receiving treatment in a control cluster, so that $Z_{kj}\sum_{i=1}^{n_{kj}}D_{kjn_{kj}}+(1-Z_{kj})\sum_{i=1}^{n_{kj}}(1-D_{kjn_{kj}}) = 0$. In such an eventuality, cluster $kj$ does not provide any information about the treatment effect and should be discarded from further analysis.
\end{remark}

Let $\mathbf{CO} = \big\{\text{CO}_{kj}, ~k = 1,\cdots, K,~j = 1, 2\big\}$ contain the $\text{CO}_{kj}$ information of all $2K$ clusters. Define the following estimand for a fixed $\mathbf{CO}$ configuration:
\begin{equation*}
    \frac{1}{2K}\sum_{k = 1}^K\sum_{j = 1}^2 \frac{\sum_{i = 1}^{n_{kj}} r_{Tkji} - \sum_{i = 1}^{n_{kj}} r_{Ckji}}{\sum_{i = 1}^{n_{kj}} d_{Tkji} - \sum_{i = 1}^{n_{kj}} d_{Ckji}} = \frac{1}{2K}\sum_{k = 1}^K\sum_{j = 1}^2 \frac{\sum_{i = 1}^{n_{kj}} r_{Tkji} - \sum_{i = 1}^{n_{kj}} r_{Ckji}}{\text{CO}_{kj}} = \lambda_{\text{ACER}}^{\mathbf{CO}}.
\end{equation*}
Consider testing the hypothesis
\begin{equation*}
    H_{0, \text{ACER}}^{\mathbf{CO}}: \lambda_{\text{ACER}}^{\mathbf{CO}} = \lambda_0.
\end{equation*}
Lemma \ref{lemma: test weak null given CO} constructs a family of randomization-based, asymptotically valid, level-$\alpha$ tests for $H_{0, \text{ACER}}^{\mathbf{CO}}$.

\begin{lemma}\rm
\label{lemma: test weak null given CO}
Let $V_k(\text{CO}_{kj}, \lambda_0) = \sum_{j = 1}^2 (2Z_{kj} - 1)\cdot \text{CO}_{kj}^{-1}\cdot\left(\sum_{i = 1}^{n_{kj}} R_{kji}\right) - \lambda_0$, $Q$ an arbitrary $K \times p$ matrix such that $p < K$, $H_Q$ the hat matrix of $Q$ with $k^{\text{th}}$ diagonal element $h_{Qk}$, $V_Q$ a column vector with entry $V_k/\sqrt{1 - h_{Qk}}$, and $\overline{V}(\text{CO}_{kj}, \lambda_0) = K^{-1}\sum_{k = 1}^K V_k(\text{CO}_{kj}, \lambda_0)$. Define the test statistic
\begin{equation*}
\delta(\lambda_0; \mathbf{CO}) = \left|\frac{\sqrt{K}\cdot\overline{V}(\text{CO}_{kj}, \lambda_0)}{ \sqrt{\frac{1}{K} V_Q^T (I-H_Q) V_Q}}\right|.
\end{equation*}
The null hypothesis $H_{0, \text{ACER}}^{\mathbf{CO}}: \lambda_{\text{ACER}}^{\mathbf{CO}} = \lambda_0$ is rejected at level $\alpha$ if $\delta(\lambda_0; \mathbf{CO}) \geq z_{1 - \alpha/2}$, the $1-\alpha/2$ quantile of the standard normal distribution.
\end{lemma}

\begin{remark}\rm
Lemma \ref{lemma: test weak null given CO} is a simple consequence of Proposition \ref{prop: CLT effect ratio} and \ref{prop: variance improvement}. As elaborated in Section \ref{subsec: effect ratio}, we recommend taking $Q = Q_1$ or $Q_2$ to fully leverage the clustered design and make variance estimator less conservative.
\end{remark}

Let $\mathcal{CO}$ denote the collection of all $\mathbf{CO}$ configurations subject to the box constraints \eqref{eqn: box constraint}. One obvious strategy to construct a valid test is to compute the minimum test statistic $\delta(\lambda_0; \mathbf{CO})$ for all $\mathbf{CO} \in \mathcal{CO}$. However, this can be unduly conservative as it ignores other useful information data suggests. Let
\[
S_{\mathbf{CO}} = \frac{1}{2K} \sum_{k = 1}^K \sum_{j = 1}^2 \iota_{kj} = \frac{1}{2K} \sum_{k = 1}^K \sum_{j = 1}^2 \frac{\text{CO}_{kj}}{n_{kj}}.
\]
If we further assume no defiers, $S_{\mathbf{CO}}$ can be understood as the average compliance rate across all $2K$ clusters. It can also be understood as a weighted L1-norm of the length-$2K$ vector $\mathbf{CO}$. Lemma \ref{lemma: construct conf for S_CO} is analogous to Lemma \ref{lemma: test weak null given CO}, and constructs a family of randomization-based, asymptotically valid, level-$\alpha/2$ confidence intervals for $S_{\mathbf{CO}}$.

\begin{lemma}\rm
\label{lemma: construct conf for S_CO}
Let $D_k = \sum_{j = 1}^2 (2Z_{kj} - 1)\cdot n_{kj}^{-1}\cdot \left(\sum_{i = 1}^{n_{kj}} D_{kji}\right)$, $Q$ an arbitrary $K \times p$ such that $p < K$, $H_Q$ the hat matrix of $Q$ with $k^{\text{th}}$ diagonal element $h_{Qk}$, $D_Q$ a column vector with entry $D_k/\sqrt{1 - h_{Qk}}$, and $\overline{D} = K^{-1}\sum_{k = 1}^K D_k$. A level-$\alpha/2$ confidence interval of $S_{\mathbf{CO}}$ is
\begin{equation*}
    I_{\alpha/2} = \left[~\overline{D} - \frac{z_{1 - \alpha/4}}{K}\cdot \sqrt{D_Q^T (I-H_Q) D_Q}, ~\overline{D} + \frac{z_{1 - \alpha/4}}{K}\cdot \sqrt{D_Q^T (I-H_Q) D_Q}~\right].
\end{equation*}
\end{lemma}
Similar to Lemma~\ref{lemma: test weak null given CO}, we recommend using a design matrix $Q$ that contains covariates predictive of $D_k$ to reduce the length of $I_{\alpha/2}$ and improve efficiency. Lemma \ref{lemma: construct conf for S_CO} implies that not all $\mathbf{CO} \in \mathcal{CO}$ are equally plausible; it turns out that it suffices to restrict our attention to $\mathbf{CO}$ configurations such that $S_{\mathbf{CO}} \in I_{\alpha/2}$. Algorithm \ref{alg} formally states the testing procedure for the null hypothesis $H_{0, \text{ACER}}: \lambda_{\text{ACER}} = \lambda_0$, and Proposition \ref{prop: validity of the test} establishes its validity.

\begin{algorithm} 
\label{algo: testing the weak null}
\SetAlgoLined
\caption{Pseudo Algorithm for Testing $H_{0, \text{ACER}}: \lambda_{\text{ACER}} = \lambda_0$ at level $\alpha$} \label{alg}
\vspace*{0.3 cm}
\KwIn {$\big\{(Z_{kj}, D_{kji}, R_{kji}), k = 1, \cdots, K, j = 1, 2, i = 1 \cdots, n_{kj}\big\}$, $\iota_{\text{min}}$, $Q^{K\times p}$, and level $\alpha$;}
\vspace*{0.3 cm}
\ShowLn Construct a level-$\alpha/2$ confidence interval for $S_{\mathbf{CO}}$ according to Lemma \ref{lemma: construct conf for S_CO}; denote it as $I_{\alpha/2}$;\\
\vspace*{0.3 cm}
\ShowLn Compute the minimum test statistic $\delta_{\text{min}}(\lambda_0) = \inf\limits_{\mathbf{CO}} \delta(\lambda_0; \mathbf{CO})$, where $\delta(\lambda_0; \mathbf{CO})$ is calculated according to Lemma \ref{lemma: test weak null given CO} and the infimum is taken over all $\mathbf{CO}$ such that $S_{\mathbf{CO}} \in I_{\alpha/2}$, and satisfies
\[
n_{kj}\cdot\iota_{\text{min}} \leq \text{CO}_{kj} \leq Z_{kj}\sum_{i=1}^{n_{kj}}D_{kjn_{kj}}+(1-Z_{kj})\sum_{i=1}^{n_{kj}}(1-D_{kjn_{kj}});
\]\\
\vspace*{0.1 cm}
\ShowLn Reject the null hypothesis $H_{0, \text{ACER}}: \lambda_{\text{ACER}} = \lambda_0$ if $\delta_{\text{min}}(\lambda_0) > z_{1-\alpha/4}$.
\end{algorithm}

\begin{proposition}\rm
\label{prop: validity of the test}
Under identification assumptions (A1), (A2), (A3') and (A4), and mild regularity conditions specified in Supplementary Material B, Algorithm 1 is an asymptotically valid level-$\alpha$ test for $H_{0, \text{ACER}}: \lambda_{\text{ACER}} = \lambda_0$ for an arbitrary $K \times p$ matrix $Q$ such that $p < K$ and $\iota_{\text{min}} > 0$. 
\end{proposition}


\subsection{Optimization and computation}
\label{subsec: algorithm}
To carry out Algorithm \ref{alg}, it is essential to compute $\delta_{\text{min}}(\lambda_0)$ subject to the box constraints, the norm constraint, and the integrality constraint. Fix a $\lambda_0$ and $\mathbf{CO}$ configuration. Lemma \ref{lemma: test weak null given CO} suggests that $\delta^2(\lambda_0; \mathbf{CO}) = (K^2 \overline{V}^2)/\big\{V^T_Q(I - H_Q)V_Q\big\}$ follows a $\chi_1^2$ distribution, and the null hypothesis $H_{0, \text{ACER}}^{\mathbf{CO}}: \lambda_{\text{ACER}}^{\mathbf{CO}} = \lambda_0$ is rejected at level $\alpha/2$ when $(K^2 \overline{V}^2)/\big\{V^T_Q(I - H_Q)V_Q\big\} \geq \chi^2_{1, \alpha/2}$. Define $\rho(\mathbf{CO})$ =  $K^2 \overline{V}^2 - \chi^2_{1, \alpha/2}\cdot V^T_Q(I - H_Q)V_Q$, where the dependence on $\mathbf{CO}$ is made explicit. Let $I_{\alpha/2} = [L_{\alpha/2}, U_{\alpha/2}]$ denote a level-$\alpha/2$ confidence interval returned from Step 1. Step 2 in Algorithm \ref{alg} reduces to solving the following optimization program:
\begin{equation}
\label{eqn: optimization prob}
\begin{split}
 &\underset{\mathbf{CO}}{\text{minimize}}~~\rho(\mathbf{CO}) \\
    \text{subject to}\quad&n_{kj}\cdot\iota_{\text{min}} \leq \text{CO}_{kj} \leq Z_{kj}\sum_{i=1}^{n_{kj}}D_{kjn_{kj}}+(1-Z_{kj})\sum_{i=1}^{n_{kj}}(1-D_{kjn_{kj}}), \\
    &L_{\alpha/2} \leq \frac{1}{2K} \sum_{k = 1}^K \sum_{j = 1}^2 \text{CO}_{kj}/n_{kj} \leq U_{\alpha/2}, ~~\text{CO}_{kj} ~\text{are integers},
\end{split}
\end{equation}
and the null hypothesis $H_{0, \text{ACER}}: \lambda_{\text{ACER}} = \lambda_0$ is rejected if the minimum returned is non-negative.

The optimization problem \eqref{eqn: optimization prob} is an instance of a mixed-integer quadratic programming (MIQP) problem, as the decision variables $\text{CO}_{kj}$ are only allowed to take on integer values. A further complication arises as decision variables $\mathbf{CO}_{kj}$ appear in the objective function as their corresponding inverses $1/\mathbf{CO}_{kj}$. To handle this, we introduce another $2K$ decision variables $\mathbf{CO}'_{kj}$ that satisfy $\mathbf{CO}_{kj} \cdot \mathbf{CO}^\prime_{kj} = 1$ for all $k,j$, which further introduces $2K$ bilinear constraints, and the optimization problem now becomes a mixed integer quadratically-constrained programming (MIQCP) problem. Supplementary Material D contains details on how to formulate and practically solve this optimization problem.

\section{Application} \label{sec: application}
Recall that we have formed $204$ pairs of $2$ similar surgeons based on the composition of their patient population, hospital characteristics, and cluster size. During the matching process, we strengthened the instrument by penalizing two surgeons having similar preference for TEE. Our final matched samples consist of similar surgeons (see Table \ref{tbl: balance table} for covariate balance) with substantially different preference for TEE. The left panel of Figure \ref{fig: boxplots} shows the boxplots of the encouragement dose in the encouraged and control groups. The median matched cluster difference in the encouragement dosage is $0.53$ and the minimum $0.23$. The right panel of Figure \ref{fig: boxplots} plots patients' 30-day mortality rate in the encouraged group of surgeons, the control group of surgeons, and the encouraged-minus-control matched cluster pair difference. The control group appears to have slightly higher 30-day mortality rate than the encouraged group. 

   \begin{figure}[h]
    \centering
    \begin{minipage}{0.49\textwidth}
        \centering
        \includegraphics[width=0.99\textwidth]{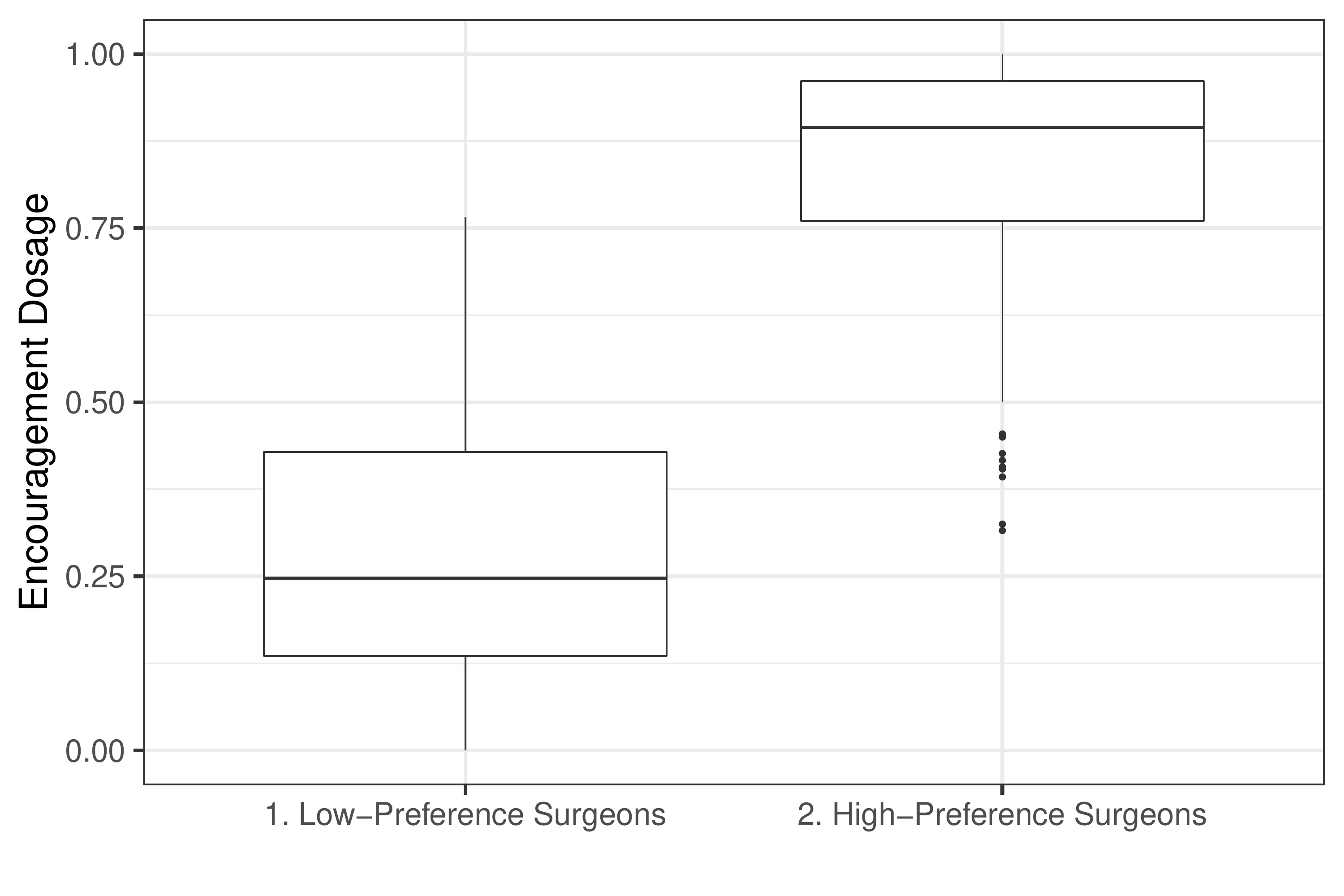} 
    \end{minipage}\hfill
    \begin{minipage}{0.49\textwidth}
        \centering
        \includegraphics[width=0.99\textwidth]{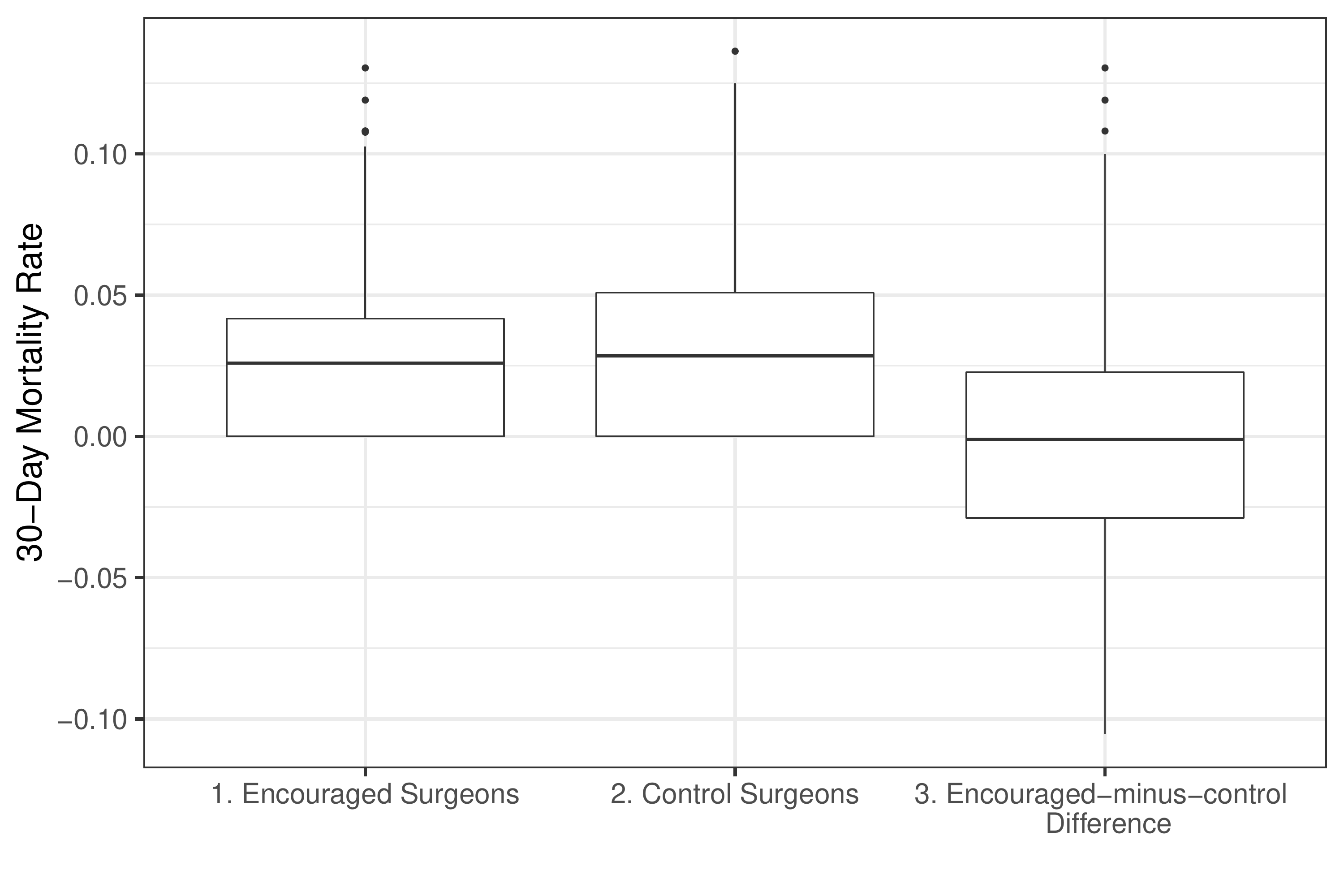} 
    \end{minipage}
    \caption{Left panel: boxplots of the encouragement dosage in the matched encouraged and control groups. Right panel: boxplots of the 30-day mortality rate in the encouraged group, 30-day mortality rate in the control group, and encouraged-minus-control matched pair difference in the 30-day mortality rate.}
    \label{fig: boxplots}
\end{figure}

Assuming a cluster-level constant proportional treatment effect model \eqref{eqn: constant prop model} and testing the sharp null hypothesis $H_{0, \text{sharp}, \text{prop}}: \beta = 0$ using a nonparametric double rank test statistic with $\varphi(d_{k}, q_{k}) = d_k \times q_k$ (see Supplementary Material C for details on the test statistic), we obtained a two-sided p-value equal to $0.045$ and a $95\%$ confidence interval $[-0.019, -0.001]$ upon inverting the test. Next, we tested the null hypothesis that the pooled effect ratio is $0$, i.e., $H_{0, \text{PER}}: \lambda_{\text{PER}} = 0$. We constructed the variance estimator according to \eqref{eqn: S^2_Q} where $Q$ is a matrix adjusting for both cluster-level and average within-cluster individual-level covariates. The two sided p-value was $0.195$ and the $95\%$ confidence interval $[-0.016, 0.003]$. Finally, we considered the cluster-heterogeneous proportional treatment effect model \eqref{eqn: cluster heterogeneous prop model} and constructed a level-0.05 confidence interval for the average cluster effect ratio. A level $0.025$ confidence interval of $\overline{\iota} = (1/2K)\sum_{k= 1}^K \sum_{j = 1}^2 \iota_{kj}$, i.e., the average compliance rate across 2K clusters if we further assume no defiers, was $[0.52, 0.58]$. As the instrument was strengthened and any two clusters in a pair had a sharp difference in the encouragement dosage, we assumed $\iota_{\text{min}} = 0.20$ and obtained a $95\%$ confidence interval $[-0.14, 0.10]$ for the average cluster effect ratio using Algorithm 1. If we allowed surgeons' preference to be an even weaker instrument at least in some clusters and set $\iota_{\text{min}} = 0.10$, we obtained a longer $95\%$ confidence interval $[-0.27, 0.24]$. It is not surprising that allowing each cluster's effect ratio to be heterogeneous as in the average cluster effect ratio estimand yields a much wider confidence interval compared to assuming a constant effect ratio across all $408$ surgeons. The long confidence interval of the average cluster effect ratio is consistent with the seemingly heterogeneous encouraged-minus-control difference in the outcome, which ranges from $0.13$ to $-0.11$ (see the third boxplot in the right panel of Figure \ref{fig: boxplots}). On the other hand, the confidence interval of the pooled effect ratio is comparable in length to assuming a constant proportional treatment effect, which makes sense in light of Remark \ref{remark: PER ACER compliance}: the pooled effect ratio assigns larger weight to surgeons with high-compliance-rate patient population and hence for whom the surgeon-specific treatment effect is the most informative; as a result, the inference for the pooled effect ratio is very efficient in this example because the average compliance rate $\overline{\iota}$ is high. 

\section*{Acknowledgement}
The authors would like to thank Professor Dylan S. Small and participants at the University of Pennsylvania causal inference reading group for helpful comments and feedback.

\section*{Supplementary Material}
Supplementary material includes a literature review, proofs to Lemma \ref{lemma: test weak null given CO}, \ref{lemma: construct conf for S_CO}, Proposition \ref{prop: CLT effect ratio}, \ref{prop: variance improvement}, and \ref{prop: validity of the test}, derivation of a class of nonparametric test statistics for the cluster-level sharp null hypothesis, formulation of Algorithm 1 into a mixed integer quadratically-constrained programming (MIQCP) optimization problem and discussion (with simulation) of its computation cost, an extensive simulation study to contrast a cluster-level instrumental variable matched analysis to a individual-level instrumental variable analysis when there is residual unmeasured confounding, and details on application.

\bibliographystyle{apalike}
\bibliography{paper-ref2}

\clearpage
\pagenumbering{arabic}
  \begin{center}
    {\LARGE\bf Supplementary Materials for ``Bridging preference-based instrumental variable studies and cluster-randomized encouragement experiments: study design, noncompliance, and average cluster effect ratio"}
\end{center}
  \medskip

\begin{abstract}
    Supplementary Material A contains a detailed literature review on approaches to cluster-randomized encouragement designs and physician-preference-based instrumental variables analysis, and statistical-matching-based methods to instrumental variable analysis. Differences between our method and those in the literature are examined. Supplementary Material B contains proofs to lemmas and propositions in the main article. Supplementary Material C derives a rich family of nonparametric test statistics for testing the cluster-level Fisher's sharp null hypothesis. Supplementary Material D introduces the standard form of  mixed integer quadratically-constrained programming (MIQCP), provides details on how to formulate the optimization problem in Algorithm 1 as an MIQCP, discusses how to practically implement the optimization problem, and reports computation cost via a simulation study. Supplementary Material E contrasts cluster-level instrumental variable matched analyses to individual-level instrumental variable matched analyses when there is residual unmeasured confounding from a causal directed acyclic graph (DAG) perspective and via extensive simulation studies. Supplementary Material F provides details for our application, including the data sources, detailed exclusion and inclusion criteria, and how statistical matching was performed.
\end{abstract}

\begin{center}
\section*{\large\bf Supplementary Material A: Literature Review}
\end{center}

Many authors have proposed methods to deal with cluster-randomized encouragement designs from different perspectives. Notably, \citet{frangakis2002clustered} proposed a Bayesian hierarchical modeling approach that estimated the intention-to-treat (ITT) effect for principal strata. \citet{forastiere2016identification} further built upon \citet{frangakis2002clustered} and proposed a Bayesian principal stratification method that disentangled the effect of encouragement on outcome via the spillover effect from that via lifting the uptake of treatment. Empirical preference-based instrumental variable analyses typically leverage structural equation models, e.g., two-stage least squares, marginal structural models, and variants of them (\citealp{brookhart2006evaluating, hernan2006instruments}). One exception is \citet{fogarty2019biased}, who conducted a matched observational study and paired individual patients whose surgeons have different preference for the treatment.

Our work drew heavily from statistical-matching-based study design approaches to instrumental variable analysis. Statistical matching is an extensively used tool to control for observed confounding variables and draw causal inference (\citealp{rubin1973matching}; \citealp{rosenbaum2002observational,rosenbaum2010design}; \citealp{hansen2004full}; \citealp{stuart2010matching}; \citealp{zubizarreta2012using}; \citealp{diamond2013genetic}; \citealp{pimentel2015large}; \citealp{savje2017generalized}; \citealp{yu2019matching}). Two key ingredients of cluster-randomized encouragement designs, clustered treatment assignment and randomized encouragement design, have been studied separately under a matching design framework in the literature. \citet{hansen2014clustered} studied testing Fisher's sharp null hypothesis in a matched observational study with cluster-level treatment assignment in non-instrumental-variable settings. They found that the clustered treatment assignment is less susceptible to hidden bias compared to the treatment assignment applied at the individual level, in the sense that if there is a genuine treatment effect and no unmeasured confounding, the clustered treatment assignment exhibits larger insensitivity to hidden bias, i.e., researchers would be able to reject the null hypothesis of no effect at a larger degree of hypothetical unmeasured confounding when conducting a sensitivity analysis. In another word, clustered treatment assignment exhibits larger design sensitivity (\citealp{rosenbaum2004design}).

\citet{small2008war} considered a matching-based study design approach to instrumental variable analysis and conducted randomization-based inference for the structural parameter in a constant proportional treatment effect model. See also \citet{imbens2005robust} and \citet{ertefaie2018quantitative}. More recently, \citet{heng2019instrumental} studied the trade-off between sample size and IV strength for the constant proportional treatment effect model. On a more practical side, \citet{zubizarreta2017optimal} and \citet{pimentel2018optimal} developed integer-programming-based and network-flow-based algorithms suited for matching clustered and multilevel data structure. 
 
Works that are most relevant to our analysis in the current article are \citet{small2008randomization}, \citet{imai2009essential}, \citet{baiocchi2010building}, \citet{kang2018estimation}, and \citet{fogarty2019biased}. \citet{small2008randomization} analyzed a group randomized trial with noncompliance, an ideal prototype of a cluster-randomized encouragement design, using randomization-based inferential methods. \citet{small2008randomization} restricted their attention to data from a cluster-randomized trial, not the observational data studied in this article. Moreover, \citet{small2008randomization} only considered binary encouragement assignment (i.e., binary instrumental variable) and conducted randomization-based inference only for Fisher's sharp null hypothesis in a constant proportional treatment effect model that does not allow for treatment heterogeneity. In the current article, we studied how to deal with continuous encouragement assignment in observational data settings, and proposed estimands and inferential methods allowing for treatment heterogeneity. 

\citet{imai2009essential} considered in great detail matched-pair cluster-randomized trials. In particular, \citet[Section 6.3]{imai2009essential} studied design-based inferential methods under individual-level noncompliance when both clusters and individuals within each cluster are randomly sampled from a superpopulation, and when only clusters are assumed to be randomly sampled from a superpopulation but individuals within each cluster are held fixed. More recently, \citet{kang2018estimation} studied randomization-based inferential methods in cluster-randomized trials with noncompliance using the so-called ``finite-sample asymptotics'' (\citealp{li2017general}). The estimand considered in \citet{kang2018estimation} is similar to the pooled effect ratio estimand considered in this article. Neither \citet{imai2009essential} nor \citet{kang2018estimation} leveraged observed covariates in constructing the variance estimator. Our work built upon \citet{imai2009essential} and \citet{kang2018estimation}, and largely expanded the scope of their work in two aspects. First, we demonstrated the usefulness of cluster-randomized encouragement experiments in the context of observational instrumental variable data with both binary and continuous instrumental variable. Second, we proposed a novel estimand (ACER) that allows each cluster to have its own cluster-level effect ratio and studied randomization-based inferential methods for this new estimand.

\citet{baiocchi2010building} first adapted the ``nonbipartite matching'' (\citealp{lu2001matching, lu2011optimal}) to instrumental variable analysis with a continuous instrument, and formally introduced the ``effect ratio'' estimand. Effect ratio estimand is an analogue of the Wald estimator in matched observational studies and allows for treatment heterogeneity. \citet{baiocchi2010building} also derived a valid randomization-based inference method for the effect ratio estimand.  \citet{fogarty2019biased} generalized inference for the effect ratio to a ``biased randomization" scheme. Our work differs from those by \citet{baiocchi2010building} and \citet{fogarty2019biased} in two aspects. First, we studied in detail how to embed noisy observational data into cluster-randomized encouragement experiments, rather than the non-clustered, individual-level, design and analysis. Second, we proposed a new model and a new estimand that generalized the constant proportional treatment effect model by allowing for a cluster-heterogeneous treatment effect. 

\clearpage
\begin{center}
\section*{{\large\bf Supplementary Material B: Proofs}}
\end{center}

\subsection*{B.1: Regularity Conditions}

We state the regularity conditions used in Section 3.2. To simplify notations, for every $k=1,\dots, K$ and $j=1,2$, we define
\[
\Delta_{Tkj}= \sum_{i=1}^{n_{kj}} (r_{Tkji}- \lambda_0 d_{Tkji}), \qquad \Delta_{Ckj}= \sum_{i=1}^{n_{kj}} (r_{Ckji}- \lambda_0 d_{Ckji}).
\]

Condition S1. (Bounded Fourth Moments)  \\
$\lim\sup_{K\rightarrow\infty} K^{-1} \sum_{k=1} ^K \Delta^4_{Tkj} $, $\lim\sup_{K\rightarrow\infty} K^{-1}\sum_{k=1} ^K \Delta^4_{Ckj} $, $j=1,2$, are finite.

Condition S2. (Existence of Population Moments) As $K\rightarrow \infty$, $K^{-1}\sum_{k=1}^K(\Delta_{Tk1}+ \Delta_{Ck1}- \Delta_{Tk2}- \Delta_{Ck2} )^2$ and  $K^{-1}\sum_{k=1}^K \tau_k^2$  converge to finite limits, where $\tau_k= 2^{-1}\sum_{j=1}^2(\Delta_{Tkj}- \Delta_{Ckj})$. 

Condition S3. (Design Matrix) $\lim\sup_{K\rightarrow\infty}K^{-1}\sum_{k=1}^K q_{kv}^4$, $v=1,\dots, p$,
are finite, where $q_{kv}$ is the $(k,v)$ element of $Q$. As $K\rightarrow\infty$, let $\bm\tau= (\tau_1, \dots, \tau_K)^T$,  $K^{-1}Q^T \bm\tau$ converges to a vector of finite limits denoted as $\bm\beta_Q$, $K^{-1} Q^T Q $ converges to a finite and positive definite limit denoted as $\Sigma_Q$, the maximum leverage  $\max_{k=1,\dots, K} h_{Qk}$ converges to zero. 

\subsection*{B.2: Proof of Proposition 1}
\begin{proof}
We first compute the expectation of $T(\lambda_0)$ under $\mathbb{E}(Z_{kj}\mid \mathcal{F}, \widetilde{\mathbf{Z}}_{\vee}, \widetilde{\mathbf{Z}}_{\wedge})=1/2$ for every $k$ and $j$:
\begin{align*}
  &\mathbb{E}\{T(\lambda_0)\mid\mathcal{F}, \widetilde{\mathbf{Z}}_{\vee}, \widetilde{\mathbf{Z}}_{\wedge}\} \\
    = &\frac{1}{K}\sum_{k=1}^{K} \mathbb{E}\left \{ \sum_{j = 1}^2 Z_{kj}\left( \sum_{i = 1}^{n_{kj}} r_{Tkji} - \lambda_{0} \sum_{i = 1}^{n_{kj}} d_{Tkji} \right)\mid\mathcal{F}, \widetilde{\mathbf{Z}}_{\vee}, \widetilde{\mathbf{Z}}_{\wedge}\right\} \\
    &\quad - \frac{1}{K}\sum_{k=1}^{K}\mathbb{E}\left\{ \sum_{j = 1}^2 (1 - Z_{kj})\left(\sum_{i = 1}^{n_{kj}} r_{Ckji} - \lambda_{0} \sum_{i = 1}^{n_{kj}} d_{Ckji} \right)\mid\mathcal{F}, \widetilde{\mathbf{Z}}_{\vee}, \widetilde{\mathbf{Z}}_{\wedge}\right\}\\
    = &\frac{1}{K}\sum_{k = 1}^K \sum_{j = 1}^2 \left\{\frac{1}{2}\sum_{i=1}^{n_{kj}} (r_{Tkji}- \lambda_0 d_{Tkji})\right\} - \frac{1}{K}\sum_{k = 1}^K \sum_{j = 1}^2 \left\{\frac{1}{2}\sum_{i=1}^{n_{kj}} (r_{Ckji}- \lambda_0 d_{Ckji})\right\}\\
    = &\frac{1}{2K} \sum_{k = 1}^K  \sum_{j=1}^{2}\sum_{i=1}^{n_{kj}}(r_{Tkji}-r_{Ckji})-\frac{1}{2K}\lambda_{0}\sum_{k = 1}^K  \sum_{j=1}^{2}\sum_{i=1}^{n_{kj}}(d_{Tkji}-d_{Ckji}).
\end{align*}
Under $H_{0, \text{PER}}: \lambda_{\text{PER}}=\lambda_0$,
\begin{align*}
    \sum_{k = 1}^K  \sum_{j=1}^{2}\sum_{i=1}^{n_{kj}}(r_{Tkji}-r_{Ckji}) = \lambda_{0}\sum_{k = 1}^K  \sum_{j=1}^{2}\sum_{i=1}^{n_{kj}}(d_{Tkji}-d_{Ckji}),
\end{align*}
and therefore 
\begin{align*}
    \mathbb{E}\{T(\lambda_0)\mid\mathcal{F}, \widetilde{\mathbf{Z}}_{\vee}, \widetilde{\mathbf{Z}}_{\wedge}\} = 0.
\end{align*}

Next, we calculate the variance:
\begin{equation*}
\begin{split}
    &\qquad {\rm var}  \{\sqrt{K}T(\lambda_0)\mid\mathcal{F}, \widetilde{\mathbf{Z}}_{\vee}, \widetilde{\mathbf{Z}}_{\wedge}\}\\
    &=\frac{1}{K}\sum_{k=1}^K {\rm var}\{Y_k(\lambda_0)\mid\mathcal{F}, \widetilde{\mathbf{Z}}_{\vee}, \widetilde{\mathbf{Z}}_{\wedge} \}\\
    &=\frac{1}{K} \sum_{k=1}^K {\rm var} \left\{\sum_{j=1}^2 (2Z_{kj}-1)\left(\sum_{i=1}^{n_{kj}}R_{kji}-\lambda_0 \sum_{i=1}^{n_{kj}} D_{kji}\right) \mid\mathcal{F}, \widetilde{\mathbf{Z}}_{\vee}, \widetilde{\mathbf{Z}}_{\wedge}\right\} \\
    &=\frac{1}{K} \sum_{k=1}^K {\rm var} \left\{ Z_{k1} \Delta_{Tk1}- (1-Z_{k1})\Delta_{Ck1}+ Z_{k2} \Delta_{Tk2}- (1-Z_{k2}) \Delta_{Ck2} \mid\mathcal{F}, \widetilde{\mathbf{Z}}_{\vee}, \widetilde{\mathbf{Z}}_{\wedge}\right\} \\
    &=\frac{1}{K} \sum_{k=1}^K {\rm var} \left\{ Z_{k1} \Delta_{Tk1}- (1-Z_{k1})\Delta_{Ck1}+ (1-Z_{k1}) \Delta_{Tk2}- Z_{k1} \Delta_{Ck2} \mid\mathcal{F}, \widetilde{\mathbf{Z}}_{\vee}, \widetilde{\mathbf{Z}}_{\wedge}\right\} \\
    &=\frac{1}{K} \sum_{k=1}^K {\rm var} \left\{ Z_{k1} (\Delta_{Tk1}+\Delta_{Ck1}- \Delta_{Tk2}- \Delta_{Ck2}) - \Delta_{Ck1}+ \Delta_{Tk2} \mid\mathcal{F}, \widetilde{\mathbf{Z}}_{\vee}, \widetilde{\mathbf{Z}}_{\wedge}\right\} \\
    &=\frac{1}{4K}\sum_{k=1}^{K}\left( \Delta_{Tk1}+ \Delta_{Ck1}- \Delta_{Tk2}- \Delta_{Ck2}\right)^2
\end{split}
\end{equation*}
where the fifth line is because $Z_{k1}+Z_{k2}=1$, the last line is because $\text{var} (Z_{k1} \mid\mathcal{F}, \widetilde{\mathbf{Z}}_{\vee}, \widetilde{\mathbf{Z}}_{\wedge})=1/4$.

Finally, under Conditions S1-S2, we prove a triangular version of the Lyapunov's condition, 
\begin{align}
    &\lim_{K\rightarrow \infty} \frac{\sum_{k=1}^K \mathbb{E}\{|Y_k(\lambda_0)- \tau_k|^4\mid \mathcal{F}, \widetilde{\mathbf{Z}}_{\vee}, \widetilde{\mathbf{Z}}_{\wedge}\}}{\{\sum_{k=1}^K  \text{var}(Y_k(\lambda_0)\mid \mathcal{F}, \widetilde{\mathbf{Z}}_{\vee}, \widetilde{\mathbf{Z}}_{\wedge} )\}^{2}}\leq 8 \lim_{K\rightarrow \infty} \frac{\sum_{k=1}^K [\mathbb{E}\{Y_k^4(\lambda_0)\mid \mathcal{F}, \widetilde{\mathbf{Z}}_{\vee}, \widetilde{\mathbf{Z}}_{\wedge}\}+ \tau_k^4]}{\{\sum_{k=1}^K  \text{var}(Y_k(\lambda_0)\mid \mathcal{F}, \widetilde{\mathbf{Z}}_{\vee}, \widetilde{\mathbf{Z}}_{\wedge} )\}^{2}}\nonumber\\
    &=8 \lim_{K\rightarrow \infty} \frac{K^{-2}\sum_{k=1}^K [\mathbb{E}\{Y_k^4(\lambda_0)\mid \mathcal{F}, \widetilde{\mathbf{Z}}_{\vee}, \widetilde{\mathbf{Z}}_{\wedge}\}+ \tau_k^4]}{\{K^{-1}\sum_{k=1}^K  \text{var}(Y_k(\lambda_0)\mid \mathcal{F}, \widetilde{\mathbf{Z}}_{\vee}, \widetilde{\mathbf{Z}}_{\wedge} )\}^{2}},\label{eq: lyapunov}
\end{align}
where $\tau_k= \mathbb{E}\{Y_k(\lambda_0)\mid \mathcal{F}, \widetilde{\mathbf{Z}}_{\vee}, \widetilde{\mathbf{Z}}_{\wedge}\}$.

Since
\begin{align}
    &\quad K^{-2} \sum_{k=1}^K \mathbb{E} \{Y_k^4(\lambda_0)\mid\mathcal{F}, \widetilde{\mathbf{Z}}_{\vee}, \widetilde{\mathbf{Z}}_{\wedge} \}\label{eq: 4moment}\\
    &= K^{-2} \sum_{k=1}^K \mathbb{E}\{[Z_{k1} (\Delta_{Tk1}- \Delta_{Ck2} ) - (1- Z_{k1}) (\Delta_{Ck1}- \Delta_{Tk2} )]^4 \mid\mathcal{F}, \widetilde{\mathbf{Z}}_{\vee}, \widetilde{\mathbf{Z}}_{\wedge} \}\nonumber \\
        &= K^{-2} \sum_{k=1}^K \mathbb{E}\{[Z_{k1} (\Delta_{Tk1}- \Delta_{Ck2} )^2 + (1- Z_{k1}) (\Delta_{Ck1}- \Delta_{Tk2} )^2]^2 \mid\mathcal{F}, \widetilde{\mathbf{Z}}_{\vee}, \widetilde{\mathbf{Z}}_{\wedge} \} \nonumber \\
         &\leq 2 K^{-2} \sum_{k=1}^K \mathbb{E}\{[Z_{k1} (\Delta_{Tk1}- \Delta_{Ck2} )^4 + (1- Z_{k1}) (\Delta_{Ck1}- \Delta_{Tk2} )^4] \mid\mathcal{F}, \widetilde{\mathbf{Z}}_{\vee}, \widetilde{\mathbf{Z}}_{\wedge} \} \nonumber \\
          &= K^{-2} \sum_{k=1}^K (\Delta_{Tk1}- \Delta_{Ck2} )^4 + (\Delta_{Ck1}- \Delta_{Tk2} )^4  \nonumber \\
    &\leq C K^{-2}  \sum_{k=1}^K \{ \Delta_{Tk1}^4+\Delta_{Tk2}^4+ \Delta_{Ck1}^4+ \Delta_{Ck2}^4 \},  \nonumber
\end{align}
where the second line is from the definition of $Y_k(\lambda_0)$, the third line is because $Z_{k1}(1-Z_{k1})=0$, the fourth line is because $(X+Y)^2\leq 2(X^2+Y^2)$ where $X, Y$ are generic variables, the fifth line is because $\mathbb{E}\{Z_{k1}\mid\mathcal{F}, \widetilde{\mathbf{Z}}_{\vee}, \widetilde{\mathbf{Z}}_{\wedge} \}=1/2 $, and $C$ in the sixth line is a generic constant. Therefore, $K^{-2} \sum_{k=1}^K \mathbb{E} \{Y_k^4(\lambda_0)\mid\mathcal{F}, \widetilde{\mathbf{Z}}_{\vee}, \widetilde{\mathbf{Z}}_{\wedge} \}\rightarrow 0$ as $K\rightarrow \infty$ from Condition S1. Moreover, 
\begin{align}
    K^{-2}\sum_{k=1}^K \tau_k^4 = 16^{-1}K^{-2}\sum_{k=1}^K \left\{ \sum_{j=1}^2 (\Delta_{Tkj}- \Delta_{Ckj}) \right\}^4 \leq CK^{-2}\sum_{k=1}^K \left( \Delta_{Tk1}^4+\Delta_{Tk2}^4+\Delta_{Ck1}^4+\Delta_{Ck2}^4 \right) \nonumber
\end{align}
where $C$ is a generic constant. Again, $K^{-2} \sum_{k=1}^K \tau_k^4\rightarrow 0$ as $K\rightarrow \infty$ from Condition S1. As a consequence, under conditions S1-S2, (\ref{eq: lyapunov}) goes to zero as $K\rightarrow \infty$. This concludes the proof of the Lyapunov's condition. 

By the central limit theorem \citep{breiman1992probability}, we have 
\begin{equation}
   \frac{ \sqrt{K} T(\lambda_0) }{\sqrt{ K^{-1}\sum_{k=1}^K {\rm var}\big\{Y_k(\lambda_0)\mid\mathcal{F}, \widetilde{\mathbf{Z}}_{\vee}, \widetilde{\mathbf{Z}}_{\wedge} \big\}}} ~\text{converges in distribution to}~ N(0,1),
\end{equation}
as $K \rightarrow \infty$.
\end{proof}

\subsection*{B.3: $S_Q^2(\lambda_0)$ is a conservative estimator for $\text{var}\{\sqrt{K}T(\lambda_0)\}$ in finite sample}
We state and prove a simple but useful lemma that says $S_Q^2(\lambda_0)$ is always a conservative estimator for $\text{var}\{\sqrt{K}T(\lambda_0)\}$ in finite sample expectation.

\begin{lemma} \label{lemma: variance expectation}
If $Q$ is invariant with respect to randomization, 
\begin{equation*}
    \mathbb{E}\big\{S^2_Q(\lambda_0)\mid\mathcal{F}, \widetilde{\mathbf{Z}}_{\vee}, \widetilde{\mathbf{Z}}_{\wedge}\big\} - {\rm var} \big\{\sqrt{K}T(\lambda_0)\mid\mathcal{F}, \widetilde{\mathbf{Z}}_{\vee}, \widetilde{\mathbf{Z}}_{\wedge}\big\} = \frac{1}{K} {\bm \mu}^T (I-H_Q) {\bm \mu}\geq 0,
\end{equation*}
where ${\bm \mu}$ is a length-K vector with each entry being $\mathbb{E}\big\{Y_k(\lambda_0)\mid\mathcal{F}, \widetilde{\mathbf{Z}}_{\vee}, \widetilde{\mathbf{Z}}_{\wedge}\big\}/\sqrt{1-h_{Qk}}$.
\end{lemma}

\begin{proof}
Recall that 
\[
S_Q^2(\lambda_0)= K^{-1} Y_Q^T (I-H_Q)Y_Q.  
\]
Define $\Lambda= {\rm var}\{Y_Q\mid\mathcal{F}, \widetilde{\mathbf{Z}}_{\vee}, \widetilde{\mathbf{Z}}_{\wedge} \}$, which is a $K\times K$ diagonal matrix with the $k^{th}$ diagonal element being ${\rm var}\big\{Y_k(\lambda_0)\mid\mathcal{F}, \widetilde{\mathbf{Z}}_{\vee}, \widetilde{\mathbf{Z}}_{\wedge} \big\}/(1-h_{Qk})$ and let $\bm\mu = \mathbb{E}\{Y_Q \mid\mathcal{F}, \widetilde{\mathbf{Z}}_{\vee}, \widetilde{\mathbf{Z}}_{\wedge}\}$. Then, 
\begin{align}
   &\quad \mathbb{E} \{K S_Q^2(\lambda_0)\mid\mathcal{F}, \widetilde{\mathbf{Z}}_{\vee}, \widetilde{\mathbf{Z}}_{\wedge} \}\nonumber \\
   &= tr[\mathbb{E}\{Y_Q^T (I-H_Q)Y_Q \mid\mathcal{F}, \widetilde{\mathbf{Z}}_{\vee}, \widetilde{\mathbf{Z}}_{\wedge}\}]\nonumber\\
   & = \mathbb{E}[tr\{Y_Q^T (I-H_Q)Y_Q \}\mid\mathcal{F}, \widetilde{\mathbf{Z}}_{\vee}, \widetilde{\mathbf{Z}}_{\wedge}]\nonumber \\
   &= \mathbb{E}[tr\{Y_Q Y_Q^T(I-H_Q)\} \mid\mathcal{F}, \widetilde{\mathbf{Z}}_{\vee}, \widetilde{\mathbf{Z}}_{\wedge}] \nonumber\\
   &=  tr[\mathbb{E}\{Y_Q Y_Q^T \mid\mathcal{F}, \widetilde{\mathbf{Z}}_{\vee}, \widetilde{\mathbf{Z}}_{\wedge}\} (I-H_Q)] \nonumber\\
   &= tr[\{\Lambda + \bm\mu \bm\mu^T\} (I-H_Q)] \nonumber\\
&= tr[\Lambda  (I-H_Q)]+  \bm\mu^T (I- H_Q) \bm\mu\nonumber\\
&= \sum_{k=1}^K {\rm var}\big\{Y_k(\lambda_0)\mid\mathcal{F}, \widetilde{\mathbf{Z}}_{\vee}, \widetilde{\mathbf{Z}}_{\wedge} \big\} +  \bm\mu^T (I- H_Q) \bm\mu, \nonumber
\end{align}
where $tr$ denotes the trace of a matrix, and the derivations mostly use the properties of trace and expectation and the condition that $I-H_Q$ is invariant with respect to randomization. Lemma \ref{lemma: variance expectation} now directly follows from $I-H_Q$ being a projection matrix and thus is positive semidefinite.

\end{proof}

\subsection*{B.4: Proof of Proposition 2}
\begin{proof}
Recall that, for every $k, j$, we defined
\begin{align}
    \Delta_{Tkj}= \sum_{i=1}^{n_{kj}} (r_{Tkji}- \lambda_0 d_{Tkji}), \qquad \Delta_{Ckj}= \sum_{i=1}^{n_{kj}} (r_{Ckji}- \lambda_0 d_{Ckji})  \nonumber.
\end{align}
Also by definition,
\begin{align}
    S^2_Q(\lambda_0)&=K^{-1} Y_Q^T(I-H_Q)Y_Q = K^{-1} Y_Q^TY_Q - K^{-1} Y_Q^T H_Q Y_Q.\label{eqn: supp S^2(Q)}
\end{align}
Consider the first term in (\ref{eqn: supp S^2(Q)}), its expectation equals  \begin{align}
    &\mathbb{E}\{K^{-1} Y_Q^TY_Q \mid \mathcal{F}, \widetilde{\mathbf{Z}}_{\vee}, \widetilde{\mathbf{Z}}_{\wedge}\} = K^{-1} \sum_{k=1}^K \frac{\mathbb{E} \{Y^2_k(\lambda_0)\mid\mathcal{F}, \widetilde{\mathbf{Z}}_{\vee}, \widetilde{\mathbf{Z}}_{\wedge} \}}{ 1-h_{Qk}} \nonumber\\
    &= K^{-1} \sum_{k=1}^K \frac{(\mathbb{E} \{Y_k(\lambda_0)\mid\mathcal{F}, \widetilde{\mathbf{Z}}_{\vee}, \widetilde{\mathbf{Z}}_{\wedge}\})^2+ {\rm var} \{Y_k(\lambda_0)\mid\mathcal{F}, \widetilde{\mathbf{Z}}_{\vee}, \widetilde{\mathbf{Z}}_{\wedge}\} }{ 1-h_{Qk}}  \nonumber  \\
    &= \frac{1}{4K} \sum_{k=1}^K \frac{\{\sum_{j=1}^2 (\Delta_{Tkj}- \Delta_{Ckj}) \}^2+ \{ \Delta_{Tk1} + \Delta_{Ck1} - \Delta_{Tk2}- \Delta_{Ck2}\}^2 }{ 1-h_{Qk}}.  \nonumber
\end{align}
Its variance equals
\begin{align}
    &{\rm var} \{ K^{-1} Y_Q^TY_Q|\mathcal{F}, \widetilde{\mathbf{Z}}_{\vee}, \widetilde{\mathbf{Z}}_{\wedge} \}= K^{-2}\sum_{k=1}^K \frac{{\rm var} \{Y_k^2(\lambda_0)\mid\mathcal{F}, \widetilde{\mathbf{Z}}_{\vee}, \widetilde{\mathbf{Z}}_{\wedge}\}}{ 1-h_{Qk}} \nonumber \\
    &\leq K^{-2} \sum_{k=1}^K \frac{\mathbb{E} \{Y_k^4(\lambda_0)\mid\mathcal{F}, \widetilde{\mathbf{Z}}_{\vee}, \widetilde{\mathbf{Z}}_{\wedge}\}}{ 1-h_{Qk}} =o(1), \nonumber
\end{align}
which is from (\ref{eq: 4moment}) and Condition S3 that $\max_{k=1,\dots, K} h_{Qk}=o(1)$. 

Therefore, by the Markov inequality, 
\begin{align}
K^{-1} Y_Q^TY_Q&= (4K)^{-1}\sum_{k=1}^K \big [ \big\{\sum_{j=1}^2 (\Delta_{Tkj}- \Delta_{Ckj}) \big\}^2+ (\Delta_{Tk1} + \Delta_{Ck1} - \Delta_{Tk2}- \Delta_{Ck2})^2 \big ]+o_p(1) \nonumber\\
&=  K^{-1}\bm\tau^T \bm\tau + {\rm var} \{ \sqrt{K} T(\lambda_0)\mid\mathcal{F}, \widetilde{\mathbf{Z}}_{\vee}, \widetilde{\mathbf{Z}}_{\wedge}\} +o_p(1). \nonumber
\end{align}
Next, consider the second term in (\ref{eqn: supp S^2(Q)}),
\begin{align}
    &K^{-1} Y_Q^T H_Q Y_Q= K^{-1} Y_Q^T Q (Q^T Q)^{-1} Q^T Y_Q = K^{-1} Y_Q^T Q (K^{-1}Q^T Q)^{-1}  K^{-1}Q^T Y_Q.  \nonumber
\end{align}
The $v$th element in $K^{-1}  Q^T Y_Q$ equals \begin{align}
    K^{-1} \sum_{k=1}^K Y_k(\lambda_0) q_{kv}/\sqrt{1-h_{Qk}}=K^{-1} \sum_{k=1}^K Y_k(\lambda_0) q_{kv}+ o(1).\nonumber
\end{align}
Notice that $\mathbb{E}\{K^{-1} \sum_{k=1}^K Y_k(\lambda_0) q_{kv}\mid  \mathcal{F}, \widetilde{\mathbf{Z}}_{\vee}, \widetilde{\mathbf{Z}}_{\wedge}\}= K^{-1} \sum_{k=1}^K \tau_k q_{kv} $, and
\begin{align}
    &\quad {\rm var }\left\{K^{-1} \sum_{k=1}^K Y_k(\lambda_0) q_{kv}\mid  \mathcal{F}, \widetilde{\mathbf{Z}}_{\vee}, \widetilde{\mathbf{Z}}_{\wedge}\right\}\nonumber\\
    &= K^{-2} \sum_{k=1}^K {\rm var} \{Y_k(\lambda_0)\mid  \mathcal{F}, \widetilde{\mathbf{Z}}_{\vee}, \widetilde{\mathbf{Z}}_{\wedge}\} q_{kv}^2\nonumber\\
    &\leq K^{-2} \sum_{k=1}^K \mathbb{E} \{Y_k^2(\lambda_0)\mid  \mathcal{F}, \widetilde{\mathbf{Z}}_{\vee}, \widetilde{\mathbf{Z}}_{\wedge}\} q_{kv}^2\nonumber\\
    &= K^{-2} \sum_{k=1}^K \mathbb{E} \left\{Z_{k1} (\Delta_{Tk1}- \Delta_{Ck2})^2+ (1-Z_{k1})(\Delta_{Ck1}- \Delta_{Tk2})^2  \mid  \mathcal{F}, \widetilde{\mathbf{Z}}_{\vee}, \widetilde{\mathbf{Z}}_{\wedge}\right\} q_{kv}^2\nonumber\\
    &= 2^{-1}K^{-2} \sum_{k=1}^K  \left\{(\Delta_{Tk1}- \Delta_{Ck2})^2+ (\Delta_{Ck1}- \Delta_{Tk2})^2 \right\} q_{kv}^2\nonumber\\
    &\leq 2^{-1}K^{-2} \left[\sum_{k=1}^K \left\{(\Delta_{Tk1}- \Delta_{Ck2})^2+ (\Delta_{Ck1}- \Delta_{Tk2})^2\right\}^2\right]^{1/2} \left[\sum_{k=1}^K q_{kv}^4\right]^{1/2}\nonumber\\
    &\leq CK^{-2} \left[\sum_{k=1}^K (\Delta_{Tk1}^4+ \Delta_{Ck2}^2+ \Delta_{Ck1}^4 + \Delta_{Tk2}^4)\right]^{1/2} \left[\sum_{k=1}^K q_{kv}^4\right]^{1/2}\nonumber\\
    &= CK^{-1} \left[K^{-1}\sum_{k=1}^K (\Delta_{Tk1}^4+ \Delta_{Ck2}^2+ \Delta_{Ck1}^4 + \Delta_{Tk2}^4)\right]^{1/2} \left[K^{-1}\sum_{k=1}^K q_{kv}^4\right]^{1/2}\nonumber\\
    &=o(1),\nonumber
\end{align}
where the fourth line uses the definition of $Y_k(\lambda_0)$, the fifth line uses $\mathbb{E}\{Z_{k1}\mid  \mathcal{F}, \widetilde{\mathbf{Z}}_{\vee}, \widetilde{\mathbf{Z}}_{\wedge}\} = 1/2 $, the sixth line uses the Cauchy-Schwartz inequality, the last line is from the bounded fourth moments in Conditions S1 and S3. Combined, we have proved that the $v$th element in $K^{-1}  Q^TY_Q$ converges in probability to the limit of $K^{-1} \sum_{k=1}^K \tau_k q_{kv}$ by the Markov inequality.  As a consequence, we have that $K^{-1} Q^T Y_Q= K^{-1} Q^T \bm{\tau}+o_p(1)$.

Finally, we have $K^{-1} Y_Q^T H_Q Y_Q= \lim_{K\rightarrow \infty} K^{-1}\bm{\tau}^T H_Q \bm{\tau} +o_p(1)$, and thus $S^2_Q(\lambda_0)- {\rm var}\{ \sqrt{K} T(\lambda_0) \mid  \mathcal{F}, \widetilde{\mathbf{Z}}_{\vee}, \widetilde{\mathbf{Z}}_{\wedge} \} =\lim_{K\rightarrow \infty}K^{-1} \bm\tau^T(I-H_Q)\bm\tau+o_p(1)$, concluding the first part of the proof.

Next, we prove the second part of the Proposition 2. We only show $S^2_{\bm{e}}(\lambda_0) - S^2_{Q_1}(\lambda_0)$ converges in probability to ${\bm \beta}_B^T \Sigma^{-1}_B{\bm \beta}_B \geq 0$ and the other statement is analogous. Let $Y_0= [ Y_1(\lambda_0), \dots, Y_K(\lambda_0)]^T$,  
\begin{align}
    &S^2_{\bm e} (\lambda_0)- S^2_{Q_1}(\lambda_0)= K^{-1} Y_0^T (I- H_{\bm e})Y_0- K^{-1} Y_0^T (I- H_{Q_1})Y_0+o(1)\nonumber\\
    &=K^{-1} Y_0^T(H_{Q_1}-H_{\bm e} )Y_0 +o(1)= K^{-1} Y_0^T(H_{\bm B}+ H_{\bm e}-H_{\bm e} )Y_0 +o(1) \nonumber\\
    &= K^{-1} Y_0^TH_{\bm B} Y_0 +o(1) = \bm\beta_B^T \Sigma_B^{-1} \bm\beta_B+o_p(1)\nonumber
\end{align}
where the second line is because $H_{Q_1}= H_{\bm B}+ H_{\bm e}$ from ${\bm B}$ being orthogonal to $\bm e$, and the last equality is directly from  $K^{-1} Y_0^T H_{\bm B} Y_0 = \lim_{K\rightarrow \infty} K^{-1}\bm{\tau}^T H_{\bm B} \bm{\tau} +o_p(1) = \bm\beta_B^T \Sigma_B^{-1} \bm\beta_B+o_p(1)$ implied by the proof of the first of the first part of Proposition 2. 
\end{proof}

\subsection*{B.5: Proof of Lemma 1}
\begin{proof}
Recall that $V_k(\text{CO}_{kj}, \lambda_0) = \sum_{j = 1}^2 (2Z_{kj}-1)\cdot \text{CO}_{kj}^{-1}\cdot\left(\sum_{i = 1}^{n_{kj}} R_{kji}\right) - \lambda_0$. We compute 
\begin{equation*}
    \begin{split}
        &\mathbb{E}\left\{\frac{1}{K} \sum_{k = 1}^K V_k(\text{CO}_{kj}, \lambda_0) \mid  \mathcal{F}, \widetilde{\mathbf{Z}}_{\vee}, \widetilde{\mathbf{Z}}_{\wedge}\right\}\\
        =& \frac{1}{K}\sum_{k = 1}^K \sum_{j = 1}^2 \mathbb{E}\left\{(2Z_{kj} - 1)\cdot \text{CO}^{-1}_{kj}\cdot \left(\sum_{i = 1}^{n_{kj}} R_{kji}\right)  \mid  \mathcal{F}, \widetilde{\mathbf{Z}}_{\vee}, \widetilde{\mathbf{Z}}_{\wedge} \right\} - \lambda_0\\
        =&\frac{1}{K}\sum_{k = 1}^K\sum_{j = 1}^2 \left\{\frac{1}{2}\cdot \text{CO}^{-1}_{kj}\cdot \left(\sum_{i = 1}^{n_{kj}} r_{Tkji}\right) - \frac{1}{2}\cdot \text{CO}^{-1}_{kj}\cdot \left(\sum_{i = 1}^{n_{kj}} r_{Ckji}\right) \right\} - \lambda_0\\
        =&\frac{1}{2K}\sum_{k = 1}^K \sum_{j = 1}^2\left\{ \text{CO}^{-1}_{kj}\cdot \left(\sum_{i = 1}^{n_{kj}} r_{Tkji} - \sum_{i = 1}^{n_{kj}} r_{Ckji}\right) \right\} - \lambda_0.
    \end{split}
\end{equation*}
Under the null hypothesis
\begin{equation}
    H_{0, \mathbf{CO}}^W: \frac{1}{2K}\sum_{k = 1}^K\sum_{j = 1}^2 \frac{\sum_{i = 1}^{n_{kj}} r_{Tkji} - \sum_{i = 1}^{n_{kj}} r_{Ckji}}{\text{CO}_{kj}} = \lambda_{\mathbf{CO}}^W = \lambda_0,
\end{equation}
we have 
\begin{equation*}
    \mathbb{E}\left\{\frac{1}{K} \sum_{k = 1}^K V_k(\text{CO}_{kj}, \lambda_0) \mid  \mathcal{F}, \widetilde{\mathbf{Z}}_{\vee}, \widetilde{\mathbf{Z}}_{\wedge}\right\} = 0.
\end{equation*}

Similar to Proposition 1, we can calculate:
\begin{equation*}
    \begin{split}
 &{\rm var} \big\{V_k(\text{CO}_{kj}, \lambda_0)\mid\mathcal{F}, \widetilde{\mathbf{Z}}_{\vee}, \widetilde{\mathbf{Z}}_{\wedge}\big\}\\
 = &4^{-1} \left\{\text{CO}^{-1}_{k1}\cdot \sum_{i=1}^{n_{k1}} (r_{Tk1i} + r_{Ck1i}) - \text{CO}^{-1}_{k2}\cdot \sum_{i=1}^{n_{k2}} (r_{Tk2i} + r_{Ck2i})\right\}^2.
    \end{split}
\end{equation*}

By the analogous triangular version of the Lyapunov's condition as in the proof of Proposition 1, we have
\begin{equation}
   \frac{ \sqrt{K}\cdot \overline{V}(\text{CO}_{kj}, \lambda_0) }{\sqrt{ K^{-1}\sum_{k=1}^K {\rm var}\big\{V_k(\text{CO}_{kj}, \lambda_0)\mid\mathcal{F}, \widetilde{\mathbf{Z}}_{\vee}, \widetilde{\mathbf{Z}}_{\wedge} \big\}}} 
   ~\text{converges in distribution to}~N(0,1),
\end{equation}
as $K \rightarrow \infty$ and under Condition S1 and S2 with $\Delta_{Tkj} = \text{CO}_{kj}^{-1}\cdot \sum_{i = 1}^{n_{kj}} r_{Tkji}$ and $\Delta_{Ckj} = \text{CO}_{kj}^{-1}\cdot \sum_{i = 1}^{n_{kj}} r_{Ckji}$

Let $Q$ an arbitrary $K \times p$ such that $p < K$, $H_Q$ the hat matrix of $Q$ with kth diagonal element $h_{Qk}$, and $V_Q$ a column vector with entry $V_k/\sqrt{1 - h_{Qk}}$. By the same argument as in the proof of Lemma \ref{lemma: variance expectation} and Proposition 2, we know that $K^{-1} V_Q^T (I-H_Q) V_Q$ is a conservative variance estimator for $\sqrt{K}\cdot \overline{V}(\text{CO}_{kj}, \lambda_0)$ both in the finite sample and asymptotically. Therefore, if we define
\begin{equation}
        \label{eqn: test statistic for H^W_0,CO}
\delta(\lambda_0; \mathbf{CO}) = \left|\frac{\sqrt{K}\cdot\overline{V}(\text{CO}_{kj}, \lambda_0)}{ \sqrt{K^{-1} V_Q^T (I-H_Q) V_Q}}\right|,
\end{equation}
the null hypothesis $H_{0, \text{ACER}}^{\mathbf{CO}}: \lambda_{\text{ACER}}^{\mathbf{CO}} = \lambda_0$ is rejected at level $\alpha$ if $\delta(\lambda_0; \mathbf{CO}) \geq z_{1 - \alpha/2}$, the $1-\alpha/2$ quantile of the standard normal distribution. This concludes the proof of Lemma 1.
\end{proof}

\subsection*{B.6: Proof of Lemma 2}

\begin{proof}
Consider testing the null hypothesis $H_{0, S_{\mathbf{CO}}}: S_{\mathbf{CO}} = \lambda_0$ using the test statistic $K^{-1} \sum_{k = 1}^K D_k - \lambda_0$, where $D_k = \sum_{j = 1}^2 (2Z_{kj} - 1)\cdot n_{kj}^{-1}\cdot \left(\sum_{i = 1}^{n_{kj}} D_{kji}\right)$. We have:
\begin{equation*}
    \begin{split}
        &\mathbb{E}\left\{\frac{1}{K} \sum_{k = 1}^K D_k(n_{k1}, n_{k2}) - \lambda_0 \mid  \mathcal{F}, \widetilde{\mathbf{Z}}_{\vee}, \widetilde{\mathbf{Z}}_{\wedge}\right\}\\
        =& \frac{1}{K}\sum_{k = 1}^K \sum_{j = 1}^2 \mathbb{E}\left\{(2Z_{kj} - 1)\cdot n^{-1}_{kj}\cdot \left(\sum_{i = 1}^{n_{kj}} D_{kji}\right)  \mid  \mathcal{F}, \widetilde{\mathbf{Z}}_{\vee}, \widetilde{\mathbf{Z}}_{\wedge} \right\} - \lambda_0\\
        =&\frac{1}{K}\sum_{k = 1}^K\sum_{j = 1}^2 \left\{\frac{1}{2}\cdot n^{-1}_{kj}\cdot \left(\sum_{i = 1}^{n_{kj}} d_{Tkji}\right) - \frac{1}{2}\cdot n^{-1}_{kj}\cdot \left(\sum_{i = 1}^{n_{kj}} d_{Ckji}\right) \right\} - \lambda_0\\
        =&\frac{1}{2K}\sum_{k = 1}^K \sum_{j = 1}^2\left\{ n^{-1}_{kj}\cdot \left(\sum_{i = 1}^{n_{kj}} d_{Tkji} - \sum_{i = 1}^{n_{kj}} d_{Ckji}\right) \right\} - \lambda_0.
    \end{split}
\end{equation*}
Recall that 
\begin{equation*}
    S_{\mathbf{CO}} = \frac{1}{2K} \sum_{k = 1}^K \sum_{j = 1}^2 \text{CO}_{kj}/n_{kj} = \frac{1}{2K}\sum_{k = 1}^K \sum_{j = 1}^2 \frac{\sum_{i = 1}^{n_{kj}}d_{Tkji} - \sum_{i = 1}^{n_{kj}}d_{Ckji}}{n_{kj}} = \lambda_0
\end{equation*}
under the null hypothesis $H_{0, S_{\mathbf{CO}}}$, Therefore, we have 
\begin{equation*}
    \mathbb{E}\left\{\frac{1}{K} \sum_{k = 1}^K D_k(n_{k1}, n_{k2}) - \lambda_0 \mid  \mathcal{F}, \widetilde{\mathbf{Z}}_{\vee}, \widetilde{\mathbf{Z}}_{\wedge}\right\} = 0
\end{equation*}
under $H_{0, S_{\mathbf{CO}}}: S_{\mathbf{CO}} = \lambda_0$. By the same arguments as in the proof of Proposition 2 and Lemma 1, we have that $K^{-1}D_Q^T (I-H_Q) D_Q$, where $Q$ an arbitrary $K \times p$ such that $p < K$, $H_Q$ the hat matrix of $Q$ with kth diagonal element $h_{Qk}$, and $D_Q$ a column vector with entry $D_k/\sqrt{1 - h_{Qk}}$, is a conservative estimator for $\sqrt{K}\cdot (\overline{D} - \lambda_0)$. Define
\begin{equation*}
    \delta_{S_{\mathbf{CO}}} = \left|\frac{\sqrt{K}(\overline{D} - \lambda_0)}{\sqrt{K^{-1} D_Q^T (I-H_Q) D_Q}}\right|.
\end{equation*}
A valid level-$\alpha/2$ test of $H_{0, S_{\mathbf{CO}}}: S_{\mathbf{CO}} = \lambda_0$ can be obtained by comparing $\delta_{S_{\mathbf{CO}}}$ to $z_{1 - \alpha/4}$. Upon inverting this test, one can construct the following valid level-$\alpha/2$ confidence interval for $S_{\mathbf{CO}}$:
\begin{equation}
    I_{\alpha/2} = \left[~\overline{D} - \frac{z_{1 - \alpha/4}}{K}\cdot \sqrt{D_Q^T (I-H_Q) D_Q}, ~\overline{D} + \frac{z_{1 - \alpha/4}}{K}\cdot \sqrt{D_Q^T (I-H_Q) D_Q}~\right].
\end{equation}
This concludes the proof of Lemma 2.
\end{proof}

\subsection*{B.7: Proof of Proposition 3}
Let $\delta_{\min}(\lambda_{0})$ be the output of Step 2 in Algorithm 1. To establish the validity of Algorithm 1, it suffices to show that for any $\mathbf{CO}\in \mathcal{CO}$, we have $ \lim_{I \rightarrow \infty} \text{pr}(\delta_{\min}(\lambda_{0})>z_{1-\alpha/4}\mid \mathcal{F}, \widetilde{\mathbf{Z}}_{\vee}, \widetilde{\mathbf{Z}}_{\wedge})\leq \alpha$. 

Note that for any $\mathbf{CO}\in \mathcal{CO}$, as $K \rightarrow \infty$,
\begin{align*}
    &\text{pr}(\delta_{\min}(\lambda_{0})>z_{1-\alpha/4}\mid \mathcal{F}, \widetilde{\mathbf{Z}}_{\vee}, \widetilde{\mathbf{Z}}_{\wedge}) \\
    = &\text{pr}(\delta_{\min}(\lambda_{0})>z_{1-\alpha/4}, S_{\mathbf{CO}}\in I_{\alpha/2} \mid  \mathcal{F}, \widetilde{\mathbf{Z}}_{\vee}, \widetilde{\mathbf{Z}}_{\wedge})
    +\text{pr}(\delta_{\min}(\lambda_{0})>z_{1-\alpha/4}, S_{\mathbf{CO}}\notin I_{\alpha/2} \mid  \mathcal{F}, \widetilde{\mathbf{Z}}_{\vee}, \widetilde{\mathbf{Z}}_{\wedge})\\
    \leq &\underbrace{\text{pr}(\delta_{\min}(\lambda_{0})>z_{1-\alpha/4}, S_{\mathbf{CO}}\in I_{\alpha/2} \mid  \mathcal{F}, \widetilde{\mathbf{Z}}_{\vee}, \widetilde{\mathbf{Z}}_{\wedge})}_{({\rm I})}  +\underbrace{\text{pr}(S_{\mathbf{CO}}\notin I_{\alpha/2} \mid \mathcal{F}, \widetilde{\mathbf{Z}}_{\vee}, \widetilde{\mathbf{Z}}_{\wedge})}_{{\rm (II)}}\\
\end{align*}    
By Lemma 2, we have ${\rm (II)} \leq \alpha/2$. 
    
Observe that    
\begin{align*}
    {\rm (I)}
    &= \text{pr}(\delta_{\min}(\lambda_{0})>z_{1-\alpha/4} \mid S_{\mathbf{CO}}\in I_{\alpha/2}, \mathcal{F}, \widetilde{\mathbf{Z}}_{\vee}, \widetilde{\mathbf{Z}}_{\wedge})
    \times \text{pr}( S_{\mathbf{CO}}\in I_{\alpha/2} \mid \mathcal{F}, \widetilde{\mathbf{Z}}_{\vee}, \widetilde{\mathbf{Z}}_{\wedge})\\
    &\leq \text{pr}(\delta_{\min}(\lambda_{0})>z_{1-\alpha/4} \mid S_{\mathbf{CO}}\in I_{\alpha/2}, \mathcal{F}, \widetilde{\mathbf{Z}}_{\vee}, \widetilde{\mathbf{Z}}_{\wedge})\\
    &\leq \alpha/2,
\end{align*}
by Lemma 1. Therefore, we have 
\begin{align*}
    \text{pr}(\delta_{\min}(\lambda_{0})>z_{1-\alpha/4}\mid \mathcal{F}, \widetilde{\mathbf{Z}}_{\vee}, \widetilde{\mathbf{Z}}_{\wedge}) = {\rm (I)} + {\rm (II)} \leq \alpha/2 + \alpha/2 = \alpha.
\end{align*}

\begin{center}
\section*{
{\large\bf Supplementary Material C: A Family of Nonparametric Test Statistics}
}
\end{center}

Recall the cluster-level sharp null hypothesis $H_{0, \text{sharp}}$ states that:
\begin{equation}\small
\begin{split}
    H_{0, \text{sharp}}: ~~&f_{kj}(r_{Tkj1}, \dots, r_{Tkjn_{kj}}, d_{Tkj1}, \dots, d_{Tkjn_{kj}})\\
    = ~&f_{kj}(r_{Ckj1}, \dots, r_{Ckjn_{kj}}, d_{Ckj1}, \dots, d_{Ckjn_{kj}}), ~\text{for all}~ k, j.
\end{split}
\end{equation}
Below, we derive a rich class of nonparametric test statistics to test the null hypothesis $H_{0, 
\text{sharp}}$. Let \[
f_{Tkj} = f_{kj}(r_{Tkj1}, \dots, r_{Tkjn_{kj}}, d_{Tkj1}, \dots, d_{Tkjn_{kj}})
\]and 
\[f_{Ckj} = f_{kj}(r_{Ckj1}, \dots, r_{Ckjn_{kj}}, d_{Ckj1}, \dots, d_{Ckjn_{kj}}).
\]
Define $\widetilde{f}_{kj} = Z_{kj} f_{Tkj} + (1 - Z_{kj})f_{Ckj}$ and $A_{k}=(Z_{k1}-Z_{k2})(\widetilde{f}_{k1} - \widetilde{f}_{k2})$. Let $(K+1)d_{k}$ be the rank of $|\widetilde{Z}_{k1}-\widetilde{Z}_{k2}|$, $(K+1)q_{k}$ the rank of $|A_{k}|$ with average ranks for ties, and $\text{sgn}(y)=1$ if $y>0$ and $0$ otherwise. Consider the class of test statistics of the form $T_{\text{DR}}=\sum_{k=1}^{K}\text{sgn}(A_{k})\cdot\varphi(d_{k}, q_{k})$,
where $\varphi: [0,1]\times [0,1] \mapsto [0, +\infty)$. The following asymptotic normality result holds under some mild conditions. 

\begin{proposition}\rm
\label{prop: double rank test}
Suppose that encouragement dose assignments are independent across matched pairs, IV assumptions A1 - A4 and $H_{0, \text{sharp}}$ hold, and \[
\lim_{K \rightarrow \infty} \frac{\max_{k=1,\dots,K}\{s_{k}\varphi(d_{k}, q_{k})\}^{2}}{\sum_{k=1}^{K}\{s_{k}\varphi(d_{k}, q_{k})\}^{2}}=0,
\] then the following asymptotic normality holds:
\begin{equation*}
    \frac{\sum_{k=1}^{K}\text{sgn}(A_{k})\varphi(d_{k}, q_{k})-\sum_{k=1}^{K}s_{k}\varphi(d_{k}, q_{k})/2}{\sqrt{\sum_{k=1}^{K}s_{k}\varphi(d_k, q_{k})^{2}/4}} ~\text{converges in distribution to}~ N(0,1),
\end{equation*}
where $s_{k}=\mathbf{1}(|\widetilde{f}_{k1}-\widetilde{f}_{k2}|>0)=\mathbf{1}(|f_{Ck1}-f_{Ck2}|>0)$ under $H_{0, \text{sharp}}$.
\end{proposition}

\begin{proof}
Note that \[
\mathbb{E}\{T_{\text{DR}}\mid \mathcal{F}, \widetilde{\mathbf{Z}}_{\vee}, \widetilde{\mathbf{Z}}_{\wedge} \}=\sum_{k=1}^{K}\mathbb{E}\{\text{sgn}(A_{k})\varphi(d_{k}, q_{k})\mid \mathcal{F}, \widetilde{\mathbf{Z}}_{\vee}, \widetilde{\mathbf{Z}}_{\wedge} \}=\sum_{k=1}^{K}s_{k}\varphi(d_{k}, q_{k})/2,
\]
and \[
\text{var}\{T_{\text{DR}}\mid \mathcal{F}, \widetilde{\mathbf{Z}}_{\vee}, \widetilde{\mathbf{Z}}_{\wedge} \}=\sum_{k=1}^{K}\text{var}\{\text{sgn}(A_{k})\varphi(d_{k}, q_{k})\mid \mathcal{F}, \widetilde{\mathbf{Z}}_{\vee}, \widetilde{\mathbf{Z}}_{\wedge} \}=\sum_{k=1}^{K}s_{k}\varphi^{2}(d_{k}, q_{k})/4.
\]
Let $B_{k}$, $k=1,2,\dots,$ be independent and identically distributed random variables with $\text{pr}(B_{k}=1)=\text{pr}(B_{k}=0)=1/2$. Note that $T_{\text{DR}}=\sum_{k=1}^{K}\text{sgn}(A_{k})\cdot\varphi(d_{k}, q_{k})$ and $\sum_{k=1}^{K}s_{k}\varphi(d_{k}, q_{k})B_{k}$ are identically distributed for all $K$. Then the desired result is an immediate consequence of Theorem 1 in the Section 6.1 of \citet{vsidak1999theory} by taking their $i$ to be $k$, their $a_{i}$ to be $s_{k}\varphi(d_{k}, q_{k})$, and their $Y_{i}$ to be $B_{k}$.
\end{proof}

Proposition~\ref{prop: double rank test} implies the following asymptotic one-sided p-value of $T_{\text{DR}}$ for testing $H_{0, \text{sharp}}$: 
\begin{equation*}
    \text{pr}(T_{\text{DR}}\geq t \mid \mathcal{F}, \widetilde{\mathbf{Z}}_{\vee}, \widetilde{\mathbf{Z}}_{\wedge})\sim 1-\Phi^{-1}\left\{\frac{t-\sum_{k=1}^{K}s_{k}\varphi(d_{k}, q_{k})}{\sqrt{\sum_{k=1}^{K}s_{k}\varphi(q_{k})^{2}/4}}\right\},~\text{as $K \rightarrow \infty$.}
\end{equation*}
The proposed family of test statistics $T_{\text{DR}}$ contains many familiar test statistics as special cases. For instance, when $\varphi(d_{k}, q_{k})=1$ for all $d_{k}$ and $q_{k}$, $T_{\text{DR}}$ reduces to the sign test. When $\varphi(d_{k}, q_{k})=q_{k}$, $T_{\text{DR}}$ reduces to the Wilcoxon signed rank test. To incorporate information contained in dose magnitude, we may consider $\varphi(d_{k}, q_{k})$ such that $\partial\varphi/\partial d_{k} \neq 0$. When the IV $\widetilde{Z}$ is ordinal, $\varphi(d_{k}, q_{k})=d_{k}q_{k}$, and $\widetilde{Z}_{k1}\wedge \widetilde{Z}_{k2}=0$ for all $k$, $T_{\text{DR}}$ reduces to the dose-weighted signed rank test considered in \citet{rosenbaum1997signed}. If $\text{cov}(d_{k}, q_{k})>0$, we tend to assign larger weights to matched pairs of clusters with larger treated-minus-control pair difference $|A_{k}|$. When $\text{cov}(d_{k}, q_{k}) = 1$, $T_{\text{DR}}$ reduces to the polynomial (with degree 2) rank test considered in \citet{rosenbaum2007confidence}. On the other hand, if $\text{cov}(d_{k}, q_{k})<0$, we are instead emphasizing more on the matched pairs with smaller $|A_{k}|$. Whether assigning more weights to larger (or smaller) $|Y_{k}|$ would make the test more powerful depends on the feature of the underlying data generating process. Let $f_{T}(x)$ denote the density function generating $f_{Tkj}$ and $f_{C}(x)$ generating $f_{Ckj}$. If the density ratio $f_{T}(x)/f_{C}(x)$ deviates from $1$ as $x$ increases, we typically prefer assigning more weights to larger $|A_{k}|$; see Section 4.3 in \citet{heng2019increasing} for more details.

\begin{center}
\section*{
{\large\bf Supplementary Material D: Details on Formulating the MIQCP Problem}
}
\end{center}

\subsection*{D.1: Mixed integer quadratically-constrained programming (MIQCP), standard form, and formulating Step 2 in Algorithm 1 as an MIQCP}
A most general mixed integer quadratically-constrained program takes on the following form (\citealp{lee2011mixed, burer2012milp})
\begin{equation}
    \begin{split}
        \label{app: MIQCP problem}
        \text{Objective:} \qquad&\text{minimize}~ \bm x^T \mathbf{Q} \bm x + \bm q^T \bm x\\
        \text{subject to:} \qquad & \bm A \bm x \leq \bm b,~~~\text{(linear constraints)} \\
        & \bm l \leq \bm x \leq \bm u, ~~~\text{(box constraints)} \\
        & \bm x^T \mathbf{Q}_i \bm x + \bm q_i^T \bm x \leq b_i, ~~~\text{(quadratic constraints)} \\
        & \text{Some or all elements of} ~\bm x~ \text{are integers.}~~~\text{(integrality constraints)}
    \end{split}
\end{equation}

Below, we transform Step 2 of Algorithm 1 into a MIQCP problem by deriving the corresponding linear, box, quadratic, and integrality constraints.

\begin{description}
\item[Decision Variables:] Recall that decision variables $\{\text{CO}_{kj}, k = 1, \cdots, K, j = 1, 2\}$ appear in the denominators in the objective function. To deal with this, we introduce another $2K$ decision variables $\{\text{CO}'_{ij}, i = 1, \cdots, K, j = 1,2\}$, which denote the corresponding inverses of $\{\text{CO}_{ij}, i = 1, \cdots, K, j = 1,2\}$, i.e., \[
\text{CO}'_{ij} = 1/\text{CO}_{ij}.
\] Finally, the decision variable $\bm x$ is the following length-$4K$ vector:
\begin{equation}
    \label{eqn: app: decision var}
    \bm x = (\text{CO}_{11}, \cdots, \text{CO}_{K1}, \text{CO}_{12}, \cdots, \text{CO}_{K2}, \text{CO}'_{11}, \cdots, \text{CO}'_{K1}, \text{CO}'_{12}, \cdots, \text{CO}'_{K2})^T.
\end{equation}

\item[Linear Constraints:] We have the following two linear constraints that restrict the norm the of decision variable vector:

\begin{equation}
    \bm a^T_1 \bm x \leq U_{\alpha/2},
    \label{eqn: first linear constraint}
\end{equation}
where 
\[
\bm a_1 = \left(\frac{1}{2K n_{11}},\cdots, \frac{1}{2K n_{K1}}, \frac{1}{2K n_{12}},\cdots, \frac{1}{2K n_{K2}}, \underbrace{ 0,\cdots,0}_{2K}\right)^T,
\]
and
\begin{equation}
    \bm a^T_2 \bm x \leq -L_{\alpha/2},
    \label{eqn: second linear constraint}
\end{equation}
where 
\[
\bm a_2 = \left(-\frac{1}{2K n_{11}},\cdots, -\frac{1}{2K n_{K1}}, -\frac{1}{2K n_{12}},\cdots, -\frac{1}{2K n_{K2}}, \underbrace{ 0,\cdots,0}_{2K}\right)^T.
\]
Observe that constraint \eqref{eqn: first linear constraint} corresponds to:
\begin{equation*}
\begin{split}
    &\bm a_1^T \bm x = \sum_{t = 1}^{4K} a_{1t}x_{t} = \sum_{k = 1}^K \frac{\text{CO}_{k1}}{2K n_{k1}} +  \sum_{k = 1}^K \frac{\text{CO}_{k2}}{2K n_{k2}} \leq U_{\alpha/2}\\
    \Longrightarrow ~&S_{\mathbf{CO}} = \frac{1}{2K}\sum_{k = 1}^K \sum_{j = 1}^2 \text{CO}_{kj}/n_{kj} \leq U_{\alpha/2}. 
\end{split}
\end{equation*}
Similarly, constraint \eqref{eqn: second linear constraint} corresponds to:
\begin{equation*}
\begin{split}
    &\bm a_2^T \bm x = \sum_{t = 1}^{4K} a_{2t}x_{t} = -\sum_{k = 1}^K \frac{\text{CO}_{k1}}{2K n_{k1}} -\sum_{k = 1}^K \frac{\text{CO}_{k2}}{2K n_{k2}} \leq -L_{\alpha/2}\\
    \Longrightarrow ~&S_{\mathbf{CO}} = \frac{1}{2K}\sum_{k = 1}^K \sum_{j = 1}^2 \text{CO}_{kj}/n_{kj} \geq L_{\alpha/2}. 
\end{split}
\end{equation*}
Together, \eqref{eqn: first linear constraint} and \eqref{eqn: second linear constraint} impose the following constraint:
\begin{equation*}
    L_{\alpha/2} \leq S_{\mathbf{CO}} \leq U_{\alpha/2}.
\end{equation*}

\item[Box Constraints:] We have the following box constraints on the first $2K$ decision variables: 
\begin{equation}
    \label{eqn: app: box}
    n_{kj}\cdot \iota_{\text{min}}\leq \text{CO}_{kj} \leq Z_{kj}\sum_{i=1}^{n_{kj}}D_{kjn_{kj}}+(1-Z_{kj})\sum_{i=1}^{n_{kj}}(1-D_{kjn_{kj}}).
\end{equation}
As elaborated in the main article, the lower bound $n_{kj}\cdot \iota_{\text{min}}$ is by Assumption (A3'), and the upper bound is an observed data quantity, which is equal to the number of individuals actually receiving treatment in the encouraged cluster, and the number of individuals not receiving treatment in the control cluster. Since $\text{CO}'_{kj} = \text{CO}_{kj}^{-1}$, we have the following box constraints on the other $2K$ decision variables:
\begin{equation}
    \label{eqn: app: box 2}
    \left\{Z_{kj}\sum_{i=1}^{n_{kj}}D_{kjn_{kj}}+(1-Z_{kj})\sum_{i=1}^{n_{kj}}(1-D_{kjn_{kj}})\right\}^{-1} \leq \text{CO}'_{kj} \leq \left\{n_{kj}\cdot \iota_{\text{min}}\right\}^{-1}.
\end{equation}

\item[Quadratic Constraints:] In order to enforce the relationship $\text{CO}_{kj}\cdot\text{CO}'_{kj} = 1$, for all $k = 1, \cdots, K, j = 1, 2$, we need the following $2K$ quadratic constraints:
\begin{equation}
    \label{eqn: app: quad constraint}
    \bm x^T \mathbf{Q}_i \bm x = 1, ~\forall i = 1, \cdots, 2K,
\end{equation}
where $\mathbf{Q}_i$ is a $4K$-by-$4K$ matrix whose $(i, i + 2K)^{\text{th}}$ and $(i + 2K, i)^{\text{th}}$ entries are $1/2$ and $0$ otherwise. Such constraints are also known as bilinear constraints in the optimization literature.

\item[Integrality Constraints:] We require that the first $2K$ decision variables take on integer values.

\item[Objective Function:] 
Recall that as shown in Section 4.3, $\delta^2(\lambda_0; \mathbf{CO}) = (K^2 \overline{V}^2)/\big\{V^T_Q(I - H_Q)V_Q\big\}$ follows a $\chi_1^2$ distribution. Our goal is to minimize 
\[
\rho(\mathbf{CO}) =  K^2 \overline{V}^2 - \chi^2_{1, \alpha/2}\cdot V^T_Q(I - H_Q)V_Q,
\]
and the null hypothesis is rejected if the minimum is non-negative. Below, we put this objective function into the standard form. For ease of exposition, we let $Q = \bm e _{K \times 1}$ be a column vector and consider the simple case where the sample variance $(K - 1)^{-1}\sum_{k = 1}^K (V_k - \overline{V})^2$ is used as a conservative variance estimator. Leveraging a general variance estimator with the help of a design matrix $Q^{K \times p}$ such that $p < K$ poses no additional difficulty. We use the classical sample variance estimator here for the simplicity of notation. Rearrange the two clusters in each pair, so that $j = 1$ corresponds to the encouraged cluster, and $j = 2$ the control cluster. To simplify notation, we 
let $R_{k, j} = \sum_{i = 1}^{n_{kj}} R_{kji}$ denote the sum of observed responses of $n_{kj}$ individuals in each cluster $kj$. Minimizing $\rho(\mathbf{CO})$ is equivalent to minimizing the following function in $\bm x_2 = (\text{CO}'_{11}, \cdots, \text{CO}'_{K1}, \text{CO}'_{12}, \cdots, \text{CO}'_{K2})$:
\begin{equation}
\begin{split}
    \label{eqn: app: objective}
    f(\bm x_2) &= \left(\frac{1 - \chi^2_{1, \alpha/2}}{K^2}\right)\cdot\sum_k (R_{k, 1}\cdot\text{CO}'_{k1} - R_{k, 2}\cdot\text{CO}'_{k2})^2 \\
    &+ \left(\frac{K - 1 + \chi^2_{1, \alpha/2}}{K^2(K - 1)}\right)\cdot \sum_{i \neq j} (R_{i, 1}\cdot\text{CO}'_{i1} - R_{i, 2}\cdot\text{CO}'_{i2})(R_{j, 1}\cdot\text{CO}'_{j1} - R_{j, 2}\cdot\text{CO}'_{j2}) \\
    &- \frac{2\lambda_0}{K} \sum_k (R_{k, 1}\cdot\text{CO}'_{k1} - R_{k, 2}\cdot\text{CO}'_{k2}) + \lambda_0^2\\
    &=\left(\frac{1 - \chi^2_{1, \alpha/2}}{K^2}\right)\left( \sum_k R^2_{k,1}\text{CO}'^2_{k1} + \sum_k R^2_{k,2} \text{CO}'^2_{k2} - 2\sum_k R_{k,1}R_{k,2}\text{CO}'_{k1}\text{CO}'_{k2} \right) \\
    &+\left(\frac{K - 1 + \chi^2_{1, \alpha/2}}{K^2(K - 1)}\right)\bigg(\sum_{i \neq j} R_{i,1}R_{j,1}\text{CO}'_{i1}\text{CO}'_{j1} + \sum_{i \neq j} R_{i,2}R_{j,2}\text{CO}'_{i2}\text{CO}'_{j2}\\ 
    &- \sum_{i \neq j} R_{i,1}R_{j,2}\text{CO}'_{i1}\text{CO}'_{j2} - \sum_{i \neq j} R_{i,2}R_{j,1}\text{CO}'_{i2}\text{CO}'_{j1}  \bigg) \\
    &- \frac{2\lambda_0}{K} \sum_k (R_{k,1}\text{CO}'_{k1} - R_{k,2}\text{CO}'_{k2}) + \lambda_0^2.
\end{split}
\end{equation}
Let $[K]$ denote $\{1, 2, \cdots, K\}$, $C_1 = (1 - \chi^2_{1, \alpha/2})/K^2$, and $C_2 = (K - 1 + \chi^2_{1, \alpha/2})/(K^2(K - 1))$. The above objective function can be written as $\bm x_2^T \mathbf{Q} \bm x_2 + \bm q^T \bm x_2$ with a $2K$-by-$2K$ matrix $\mathbf{Q}$ whose $(ij)^{th}$ entry satisfies:
\begin{equation}
\label{eqn: core Q matrix}
    Q_{ij} = \begin{cases}
    C_1 \cdot R^2_{i, 1}, \qquad &\text{for} ~i = j \in [K],\\
    C_1 \cdot R^2_{i - K, 2}, &\text{for} ~i = j \in [K] + K,\\
    -C_1\cdot R_{i, 1} R_{j, 2}, &\text{for}~i \in [K], ~j = i + K, ~\text{and}~ j \in [K], ~i = j + K,\\
    C_2\cdot R_{i, 1}R_{j, 1}, &\text{for}~i \in [K], ~j \in [K], ~i \neq j,\\
    C_2\cdot R_{i - K, 2}R_{j - K, 2}, &\text{for}~i \in [K] + K, ~j \in [K] + K, ~i \neq j,\\
    -C_2\cdot R_{i, 1}R_{j - K, 2}, &\text{for}~i \in [K], ~j \in [K] + K,\\
    -C_2\cdot R_{i - K, 2}R_{j, 1}, &\text{for}~i \in [K] + K, ~j \in [K],\\
    \end{cases}
\end{equation}
and 
\begin{equation}
    \label{eqn: app: linear in obj}
    \bm q = -\frac{2\lambda_0}{K}\left(R_{1,1}, \cdots, R_{K, 1}, -R_{1,2}, \cdots, -R_{K,2}\right)^T.
\end{equation}
Therefore, the final $\bm Q$ matrix in the objective function is equal to 
\begin{equation*}
    \bm Q_{\text{final}} = \bm Q_2^T \cdot \bm Q \cdot \bm Q_2,
\end{equation*}
where $\bm Q$ is defined as in \eqref{eqn: core Q matrix} and $\bm Q_2 = [\mathbf{O}_{2K \times 2K}, \textbf{I}_{2K \times 2K}]$, and the final $\bm q$ vector is equal to 
\begin{equation*}
    \bm q_\text{final} = \left( \underbrace{ 0,\cdots,0}_{2K}, -\frac{2\lambda_0 R_{1,1}}{K}, \cdots, -\frac{2\lambda_0 R_{K,1}}{K}, -\frac{2\lambda_0 R_{1,2}}{K}, \cdots,
    -\frac{2\lambda_0 R_{K,2}}{K}\right).
\end{equation*}
Let $\rho^\ast$ denote the minimum of $\bm x^T \bm Q_{\text{final}}\bm x + \bm q^T_{\text{final}}\bm x$. The test is now rejected at level $\alpha$ if $\rho^\ast \geq -\lambda_0^2$, and is failed to be rejected otherwise.
\end{description}

\subsection*{D.2: Solving the MIQCP problem}

Mixed integer quadratic programming (MIQP), including MIQCP, problems are well-known to be NP-hard (\citealp{burer2012milp}); however, two specific features of our problem make it feasible to solve this challenging optimization problem for practically-relevant number of clusters. First, the boundedness condition on the decision variables is crucial in a complexity sense because \citet{jeroslow1973there} showed the undecidability of unbounded problems. It is also useful in practice as it largely reduces the computational cost. Second, although the testing problem is recast as an optimization problem, our real interest is not the optimal objective function value, but rather its sign: the test is rejected if $\min \rho(\mathbf{CO}) > 0$ and fails to be rejected otherwise. Optimization routines typically maintain a lower and upper bound on the optimal objective value, where the upper bound denotes the best known feasible solution and the true optimal objective value is at least as large as the lower bound. Therefore, it suffices to terminate the optimization routine as soon as the upper bound is less than $0$, and declare that we fail to reject the null hypothesis $H_{0, \text{ACER}}: \lambda_{\text{ACER}} = \lambda_0$, or the lower bound is larger than $0$, and declare that the null hypothesis is rejected at level $\alpha$.

Optimization routines that can solve MIQCP problems include the \textsf{IBM} CPLEX Optimizer and \textsf{Gurobi} (version 9.0.1), both of which are made freely available for academic purposes. We implemented the testing procedure described in Algorithm 1 in the main article using Gurobi via its \textsf{R} interface, and made it available via our \textsf{R} package \textsf{ivdesign}.

\subsection*{D.3: Simulation Studies}
We report computation time of Algorithm 1 in this simulation study. We consider the following data generating process. Suppose there are $K = 50, 75, \cdots, 250$ pairs of $2$ clusters, each with $n_{kj} = 10$ individuals. Each individual in the cluster has a $0.25$ probability to be a complier ($\iota_C = 0.25$), $0.25$ probability to be an always-taker ($\iota_A = 0.25$), and $0.5$ probability to be a never-taker ($\iota_N = 0.5$). Hence, number of compliers, always-takers, and never-takers in each cluster $kj$, i.e., $(\text{CO}_{kj}, \text{AT}_{kj}, \text{NT}_{kj})$, is a realization of Multinomial(10, 0.25, 0.25, 0.5) distribution. We then compute $(\sum_{i = 1}^{n_{kj}} d_{Tkji}, \sum_{i = 1}^{n_{kj}} d_{Ckji})$ from $(\text{CO}_{kj}, \text{AT}_{kj}, \text{NT}_{kj})$, i.e., $\sum_{i = 1}^{n_{kj}} d_{Tkji} = \text{CO}_{kj} + \text{AT}_{kj}$ and $\sum_{i = 1}^{n_{kj}} d_{Ckji} = \text{AT}_{kj}$. The cluster-specific proportional treatment effect is simulated according to $\beta_{kj} \sim \text{Normal}(\beta, 1)$ distribution. Finally, we simulate $\sum_{i = 1}^{n_{kj}} r_{Ckji} \sim \text{Normal}(0, 1)$, and compute $\sum_{i = 1}^{n_{kj}} r_{Tkji} = \sum_{i = 1}^{n_{kj}} r_{Ckji} + \beta_{kj}(\sum_{i = 1}^{n_{kj}} d_{Tkji} - \sum_{i = 1}^{n_{kj}} d_{Ckji})$. Thus, we have generated four potential outcomes $(\sum_{i = 1}^{n_{kj}} d_{Tkji}, \sum_{i = 1}^{n_{kj}} d_{Ckji}, \sum_{i = 1}^{n_{kj}} r_{Tkji}, \sum_{i = 1}^{n_{kj}} r_{Ckji})$ for each cluster $kj$. Next, we randomly pick one cluster in each pair of two clusters to be the encouraged cluster, and the other the control cluster. The observed data consists of $(\sum_{i = 1}^{n_{kj}} d_{Tkji}, \sum_{i = 1}^{n_{kj}} r_{Tkji})$ of the encouraged cluster, and $(\sum_{i = 1}^{n_{kj}} d_{Ckji}, \sum_{i = 1}^{n_{kj}} r_{Ckji})$ of the control cluster.

\begin{figure}[h]
    \centering
    \includegraphics[width=\textwidth, height = 9.5cm]{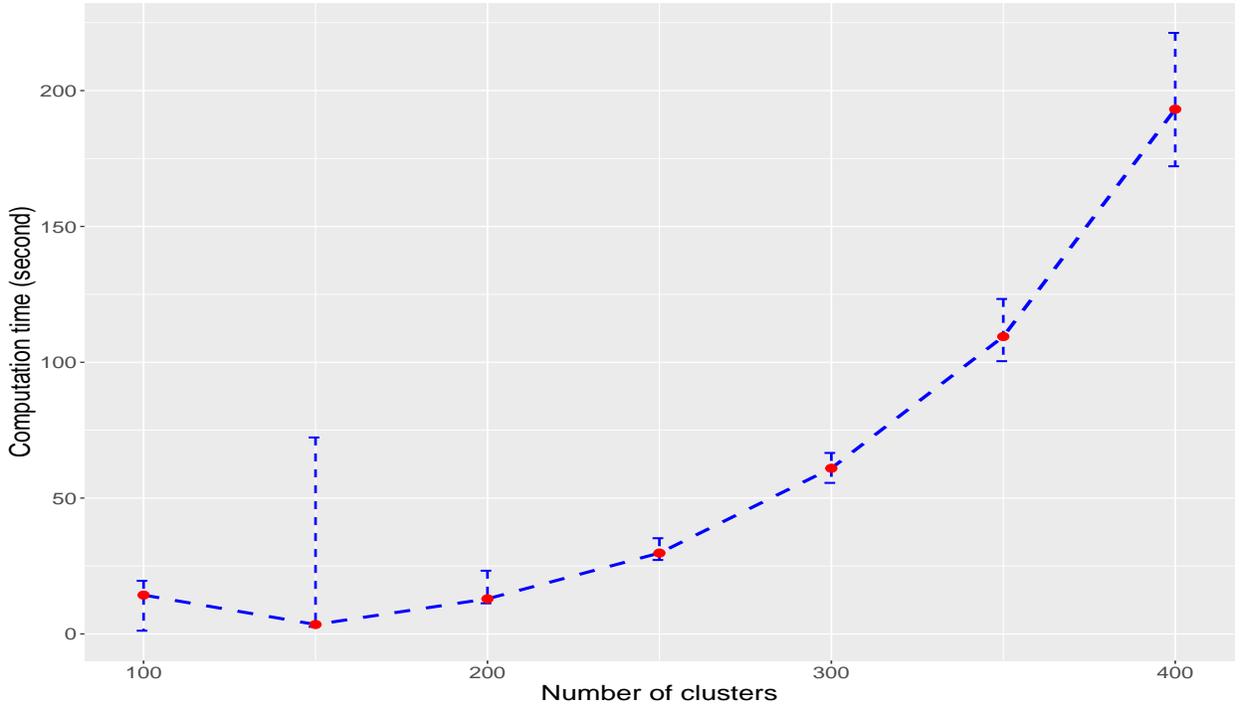}
    \caption{Computation time against the number of clusters when $\beta = 1.5$. Each red point represents the median computation time in seconds. Associated vertical intervals represent the interquartile range.}
    \label{fig: run time}
\end{figure}

\newpage

\begin{center}
{\large\bf Supplementary Material E: Sensitivity to Unmeasured Confounding}
\end{center}
\citet{hansen2014clustered} studied and contrasted the behaviors of clustered and non-clustered treatment assignment in a Rosenbaum-bounds-style sensitivity analysis in a favorable situation, i.e., a stochastic data generating process where there is a genuine treatment effect and no unmeasured confounding. In the rest of this section, we are concerned with the robustness of the primary analysis with clustered and non-clustered treatment assignment when there is residual unmeasured confounding. In other words, \citet{hansen2014clustered} are concerned with studies' robustness in a sensitivity analysis in a favorable situation, while we are concerned with primary analyses' robustness to unmeasured confounding in a non-favorable situation.

Following the notation in the main article, we let $\widetilde{Z}_k, k = 1, \cdots, 2I$ denote a continuous cluster-level instrumental variable, $\widetilde{\bm x}_k$ and $\bm x_{ki}$ cluster-level observed confounding variables and individual-level observed confounding variables, respectively, and $u_{ki}$ a individual-level unmeasured confounding variable. Note that here we are discussing and will be modeling the data generating process before pairing $2I$ clusters into $I$ pairs; therefore, we do not have subscript $j$ here and subscript $k$ ranges from $1$ to $2I$. 

The left panel of eFigure \ref{fig: two DAGs} plots a standard instrumental variable causal directed acyclic graph (DAG) where $\widetilde{Z}_k$ is a valid instrumental variable at both the individual level and cluster level. The right panel of eFigure \ref{fig: two DAGs} builds upon the standard instrumental variable directed acyclic graph by adding a common ancestor $a_k$ for $(\widetilde{\bm x}_k, \bm x_{ki})$, $u_{ki}$, and $\widetilde{Z}_k$. Here $a_k$ can be understood as a latent cluster-specific factor that partly generates observed confounding variables $(\widetilde{\bm x}_k, \bm x_{ki})$, unobserved confounding variable $u_{ki}$, and the cluster-level instrumental variable $\widetilde{Z}_k$. Observe that $a_k$ is not necessarily the sole cause of $(\widetilde{\bm x}_k, \bm x_{ki})$, $u_{ki}$, and $\widetilde{Z}_k$, as there may still be association between these variables after controlling for $a_k$. 

\begin{figure}[h]
    \centering
    \begin{minipage}{0.49\textwidth}
        \centering
         \begin{tikzpicture}
    \node[state] (1) {$\widetilde{Z}_k$};
    \node[state] (2) [right =of 1] {$D_{ki}$};
    \node[state] (3) [right =of 2] {$R_{ki}$};
    \node[state] (4) [above =of 1,xshift=1cm,yshift=-0.3cm] {$\widetilde{\bm x}_k, \bm x_{ki}$};
    \node[state] (5) [right =of 4] {$u_{ki}$};

    \path[bidirected] (1) edge node[above] {} (2);
    \path (2) edge node[above] {} (3);
    \path (4) edge node[el,above] {} (1);
    \path (4) edge node[el,above] {} (2);
    \path (4) edge node[el,above] {} (3);
    \path (5) edge node[el,above] {} (2);
    \path (5) edge node[el,above] {} (3);
\end{tikzpicture}
     \end{minipage}\hfill
    \begin{minipage}{0.49\textwidth}
        \centering
          \begin{tikzpicture}
    \node[state] (1) {$\widetilde{Z}_k$};
    \node[state] (2) [right =of 1] {$D_{ki}$};
    \node[state] (3) [right =of 2] {$R_{ki}$};
    \node[state] (4) [above =of 1,xshift=1cm,yshift=-0.3cm] {$\widetilde{\bm x}_k, \bm x_{ki}$};
    \node[state] (5) [right =of 4] {$u_{ki}$};
    \node[state] (6) [above =of 4, xshift = 1cm, yshift = -0.3cm] {$a_k$};

    \path[bidirected] (1) edge node[above] {} (2);
    \path (2) edge node[above] {} (3);
    \path (6) edge[bend right=60] node[above] {} (1);
    \path (4) edge node[el,above] {} (2);
    \path (4) edge node[el,above] {} (3);
    \path (5) edge node[el,above] {} (2);
    \path (5) edge node[el,above] {} (3);
    \path (6) edge node[el,above] {} (4);
    \path (6) edge node[el,above] {} (5);
    \path[bidirected] (4) edge node[above] {} (5);
    \path[bidirected] (4) edge node[above] {} (1);
\end{tikzpicture}
    \end{minipage}
    \caption{Left panel: causal directed acyclic graph (DAG) for a valid instrumental variable. Right panel: causal DAG for an invalid instrumental variable at individual-level. $\widetilde{\bm x}_k$ and $\bm x_{ki}$ represent cluster-level and individual-level observed confounding variables, respectively. $u_{ki}$ is an unmeasured individual-level confounding variable. $a_k$ is a cluster-level latent factor.}
    \label{fig: two DAGs}
\end{figure}
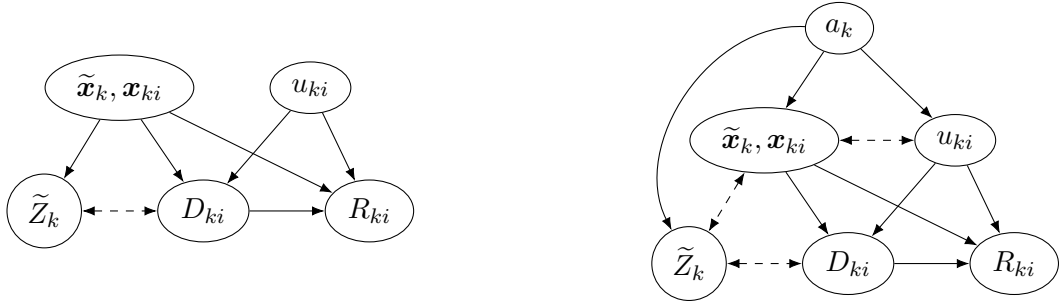



We consider the following system of simple linear structural equation models with one observed confounder corresponding to the right panel of eFigure \ref{fig: two DAGs}:
\begin{equation}
    \begin{split}
        & a_k \sim N(\mu, \sigma^2), \\
        & x_{ki} = a_k + e_{ki},\\
        &u_{ki} = a_k + f_{ki},\\
        &\begin{pmatrix}
        e_{ki} \\
        f_{ki}
        \end{pmatrix}\sim N\left(\begin{pmatrix}
        0 \\
        0
        \end{pmatrix},\begin{pmatrix}
        \sigma^2_e &~~ \rho\sigma_e \sigma_u \\
        \rho\sigma_e\sigma_u &~~ \sigma_u^2
        \end{pmatrix}\right), \\
        &\widetilde{Z}_{ki} = \widetilde{Z}_k = \gamma_0 + \gamma_1 a_k + h_k, ~h_k \sim N(0, 1), \\
        &D_{ki} = \eta_0 + \eta_1\widetilde{Z}_k + \eta_2 x_{ki} + \eta_3 u_{ki} + v_{ki},~v_{ki} \sim N(0, 1),\\
        &R_{ki} = \beta D_{ki} + \beta_1 x_{ki} + \beta_2 u_{ki} + \epsilon_{ki},~\epsilon_{ki} \sim N(0, 1).\\
    \end{split}
    \label{eq: DGP}
\end{equation}
The bias of a individual-level instrumental variable analysis is driven by the covariance between the instrumental variable $\widetilde{Z}_{ki} = \widetilde{Z}_k$ and individual-level unmeasured confounder $u_{ki}$ conditional on observed confounder $x_{ki}$. According to \eqref{eq: DGP}, we have
\begin{equation*}
    \text{cov}\{\widetilde{Z}_{ki}, u_{ki} \mid x_{ki}\} = \gamma_1 \text{var}\{a_k \mid x_{ki}\} =\gamma_1 \times  \frac{1}{\frac{1}{\sigma_e^2} + \frac{1}{\sigma^2}},
\end{equation*}
where the last equality is a direct consequence of normal-normal hierarchical structure underpinning the $(a_k, x_{ki})$ relationship (\citealp{gelman2013bayesian}). At the cluster level, the system of structural equations becomes
\begin{equation}
    \begin{split}
        &\widetilde{Z}_{k} = \gamma_0 + \gamma_1 a_k + h_k, \\
        &\overline{D}_{k} = \eta_0 + \eta_1\widetilde{Z}_k + \eta_2 \overline{x}_{k} + \eta_3 \overline{u}_{k} + \overline{v}_{k}, \\
        &\overline{R}_{k} = \beta \overline{D}_{k} + \beta_1 \overline{x}_{k} + \beta_2 \overline{u}_{k} + \overline{\epsilon}_{k}, \\
\end{split}
\end{equation}
where $\overline{x}_{k} = n^{-1}_{k} \sum_{i = 1}^{n_k} x_{ki}$, $\overline{D}_{k} = n^{-1}_{k} \sum_{i = 1}^{n_k} D_{ki}$, $\overline{u}_{k} = n^{-1}_{k} \sum_{i = 1}^{n_k} u_{ki}$, $\overline{v}_{k} = n^{-1}_{k} \sum_{i = 1}^{n_k} v_{ki}$, $\overline{R}_{k} = n^{-1}_{k} \sum_{i = 1}^{n_k} R_{ki}$, $\overline{\epsilon}_{k} = n^{-1}_{k} \sum_{i = 1}^{n_k} \epsilon_{ki}$. The bias in a cluster-level instrumental variable analysis is driven by 
\begin{equation}
    \text{cov}\{\widetilde{Z}_{k}, \overline{u}_{k} \mid \overline{x}_{k}\} = \gamma_1 \text{var}\{a_k \mid \overline{x}_k\} = \gamma_1 \times \frac{1}{\frac{1}{\sigma_e^2/n_k} + \frac{1}{\sigma^2}}.
    \label{eqn: cov cluster level}
\end{equation}

When $\sigma_e^2/n_k \ll 1$, i.e., when the number of covariates in the cluster $n_k$ is large, $\text{cov}\{\widetilde{Z}_{k}, \overline{u}_{k} \mid \overline{x}_{k}\} \ll \text{cov}\{\widetilde{Z}_{ki}, u_{ki} \mid x_{ki}\}$, and although both the individual-level and cluster-level analyses are biased, the cluster-level analysis is less biased. 

From the directed acyclic graph perspective, conditioning on observed confounding variables $(\widetilde{\bm x}_k, \bm x_{ki})$ blocks the path from $\widetilde{Z}_{k}$ to $u_{ki}$ via $(\widetilde{\bm x}_k, \bm x_{ki})$; there is still a path from $\widetilde{Z}_k$ to $u_{ki}$ via $a_k$ and hence the $\widetilde{Z}_k$ is a biased instrumental variable. In a cluster-level analysis, if $(\widetilde{\bm x}_k, \overline{x}_k)$ contains abundant information about $a_k$, conditioning on $(\widetilde{\bm x}_k, \overline{x}_k)$ amounts to also conditioning on $a_k$, and thus blocking the path from $\widetilde{Z}_k$ to $\overline{u}_{k}$ via $a_k$. For instance, in the data generating process \eqref{eq: DGP}, $\overline{x}_k \sim N(a_k, \sigma_e^2/n_k)$ and when $\sigma_e^2/n_k$ is small, $\overline{x}_k \approx a_k$, and conditioning on $\overline{x}_k$ is approximately like conditioning on $a_k$ in the sense that the information of $a_k$ is essentially contained in $\overline{x}_k$. In this way, $\widetilde{Z}_k$ becomes much less correlated with the unmeasured confounder in a cluster-level analysis compared to an individual-level analysis, and the cluster-level instrumental variable analysis becomes less biased. We finally illustrate this point via a simulation study.

We considered the simple data-generating process described in \eqref{eq: DGP} with the following parameter values: $\mu = 0$, $\sigma = \sigma_e = \sigma_u = 1$, $\rho = 0.5$, $\gamma_0 = \eta_0 = 0$, $\gamma_1 = \eta_1 = \eta_2 = \beta = \beta_1 = 1$, and $\beta_2 = 2$. For various number of clusters $K$, and number of individuals in each cluster $n_k$, we conducted both a individual-level instrumental variable matched analysis and a cluster-level instrumental variable matched analysis, and constructed $95\%$ two-sided confidence intervals for the structural parameter $\beta$ by inverting Huber's M test statistic assuming a constant additive treatment effect. For each $(K, n_k)$ combination, eTable \ref{tbl: compare unit and cluster} reports the average left and right endpoints of $95\%$ confidence intervals, average length of $95\%$ confidence intervals, and percentage of times confidence intervals cover the true parameter $\beta = 1$ when comparing the individual-level and cluster-level analysis. eTable \ref{tbl: compare unit and cluster} suggested that the cluster-level analysis was in general less efficient compared to the individual-level analysis as the length of CI was typically longer for a cluster-level analysis. Cluster-level analysis also seemed to be much less biased compared to the individual-level analysis on the same dataset. For instance, when $K = 500$ and $n_k = 50$, the average midpoint of $95\%$ CIs for the cluster-level analysis was $1.04$, while that of the individual-level analysis was $1.29$. We also observed that for a fixed $K$, the cluster-level analysis seemed to be less biased as $n_k$ grew larger, a trend consistent with \eqref{eqn: cov cluster level}. eFigure \ref{fig: CI plot} further plots the first $100$ confidence intervals for the individual-level and cluster-level analysis when $K = 500$ and $n_k = 50$. It is evident that cluster-level analysis yields slightly longer but considerably less biased confidence intervals compared to individual-level analysis.

\begin{table}[ht]
\centering
\caption{Comparing individual-level and cluster-level analysis in the presence of unmeasured confounding. $K$: number of cluster; $n_k$: number of individuals in each cluster; CI Left: average left endpoint of the $95\%$ CI; CI Right: average right endpoint of the $95\%$ CI; CI Length: average CI length; Coverage: percentage of $95\%$ CIs covering the true parameter $\beta = 1$. Simulations are repeated $500$ times.}
\label{tbl: compare unit and cluster}
\begin{tabular}{cccccccccccc}
& &\multicolumn{4}{c}{Individual Level}&&\multicolumn{4}{c}{Cluster Level} \\
 K & $n_k$ & CI Left & CI Right & CI Length &Coverage && CI Left & CI Right & CI Length &Coverage\\ 
200 & 10 & 1.17 & 1.39 & 0.21 & 0.00 && 0.91 & 1.30 & 0.38 & 0.75 \\ 
  200 & 30 & 1.23 & 1.34 & 0.12 & 0.00 && 0.93 & 1.22 & 0.29 & 0.76 \\ 
  200 & 50 & 1.24 & 1.33 & 0.09 & 0.00 && 0.93 & 1.20 & 0.27 & 0.78 \\ 
  300 & 10 & 1.20 & 1.37 & 0.17 & 0.00 && 0.96 & 1.26 & 0.30 & 0.62 \\ 
  300 & 30 & 1.24 & 1.33 & 0.10 & 0.00 && 0.95 & 1.17 & 0.22 & 0.71 \\ 
  300 & 50 & 1.25 & 1.32 & 0.07 & 0.00 && 0.96 & 1.16 & 0.19 & 0.69 \\ 
  500 & 10 & 1.22 & 1.35 & 0.13 & 0.00 && 0.99 & 1.21 & 0.22 & 0.52 \\ 
  500 & 30 & 1.25 & 1.32 & 0.07 & 0.00 && 0.97 & 1.12 & 0.15 & 0.67 \\ 
  500 & 50 & 1.26 & 1.31 & 0.05 & 0.00 && 0.98 & 1.10 & 0.12 & 0.67 \\ 
\end{tabular}
\end{table}

\begin{figure}[h]
    \centering
    \includegraphics[width = \textwidth, height = 10cm]{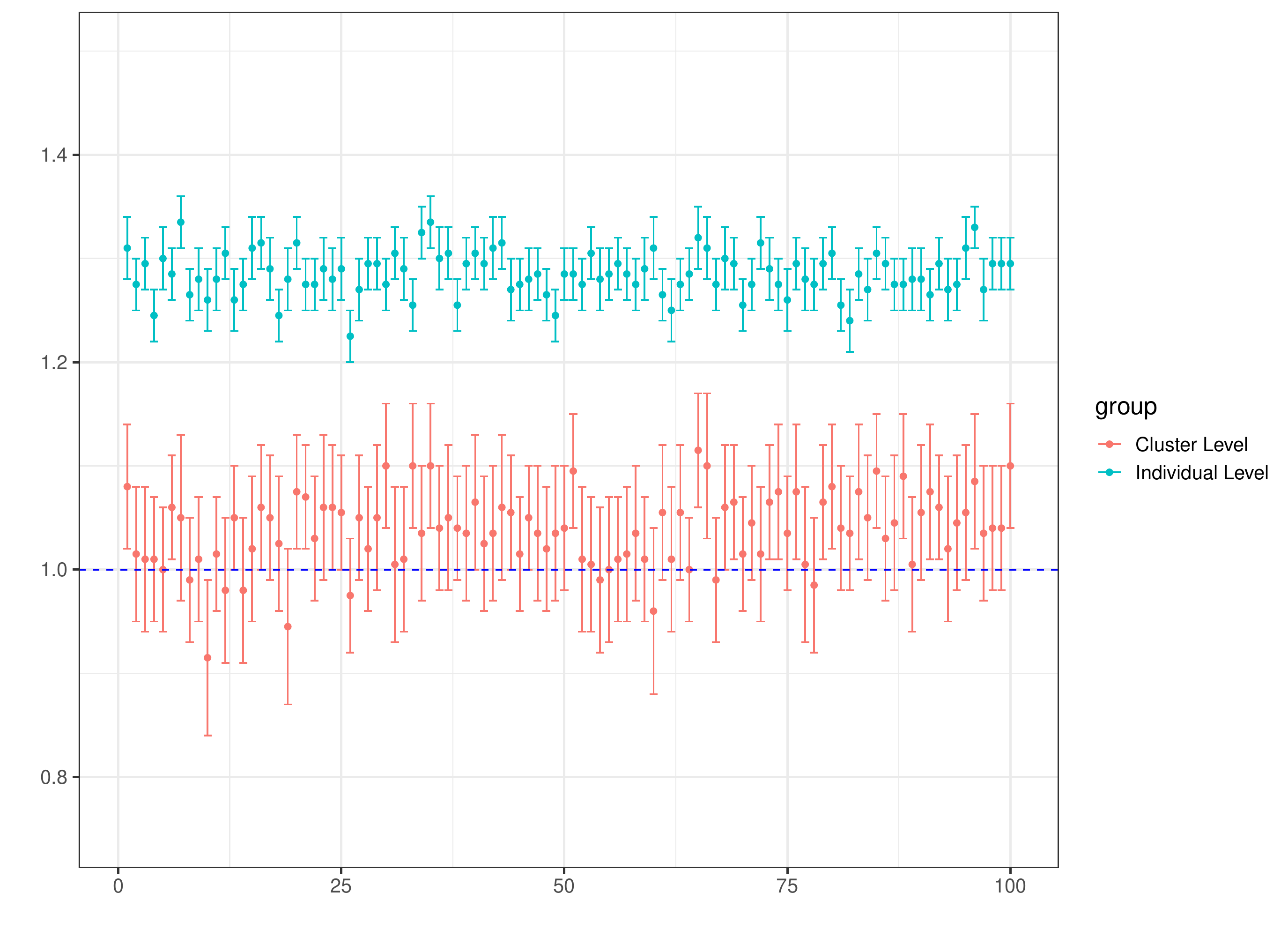}
    \caption{First $100$ $95\%$ confidence intervals for the individual-level and cluster-level analysis when $K = 500$ and $n_k = 50$. The blue dashed line represents the true parameter value $\beta = 1$.}
    \label{fig: CI plot}
\end{figure}

\clearpage

\begin{center}
{\large\bf Supplementary Material F: Application}
\end{center}

\subsection*{F.1: Data Sources and Population}
We obtained data on patients undergoing isolated CABG surgery between January 1, 2013, and October 15, 2015, from Centers for Medicare and Medicaid Services (CMS). Our dataset contained National Provider Identifier (NPI) numbers from which we identified patients' hospitals. A large subset of our dataset further contained surgeon identifier numbers from which we identified patients' surgeons. We obtained hospital characteristics data from the American Hospital Association Survey. Patient-level data were merged to hospital characteristics data using unique NPI numbers. The study cohort consisted of all fee-for-service Medicare beneficiaries with a Part A (hospitalization) Medicare claim for isolated CABG surgery. We excluded (1) beneficiaries enrolled under managed care and not fee-for-service, (2) beneficiaries with less than six months of continuous enrollment in Medicare prior to the index admission for CABG surgery, (3) beneficiaries with age $< 65$ years, (4) beneficiaries without a cardiovascular or cardiac surgery-related Diagnosis Related Group (DRG) codes, (5) beneficiaries with a neurologic or stroke diagnosis as indicated by an ICD-9-cm code within the six months prior to the index admission or a stroke diagnosis with a “present on admission” (POA) indicator. We grouped patients according to their surgeon identification number and did not include surgeons who performed less than $30$ surgeries on our study cohort.


\subsection*{F.2: Statistical Matching}
We used the \textsf{R} package \textsf{nbpMatching} (\citealp{lu2011optimal}, \citealp{R_pkg_nbpmatching}) to pair one surgeon in the New York Metropolitan area who preferred using TEE during CABG surgery to one who did not prefer using TEE during CABG surgery. Each surgeon's preference to using TEE during CABG surgery was measured by the fraction of her CABG surgeries using TEE monitoring.  Two surgeons were judged similar by the matching algorithm if they were similar in patient composition including patients' average age, percentage of male patients, percentage of white patients, percentage of elective CABG surgeries, percentage of patients having each of the following important comorbid conditions: arrhythmia, diabetes, congestive heart failure (CHF), hypertension, obesity, pulmonary diseases, and renal diseases, and similar in hospital characteristics including total hospitals beds, teaching status, presence of any cardiac intensive care unit, total number of full-time registered nurses, and total cardiac surgical volume. We calculated the distance matrix using a rank-based robust Mahalanobis distance. A large penalty $\lambda$ is further added to the $ij^{th}$ entry of this distance matrix if the absolute difference between surgeon $i$ and surgeon $j$'s preference for TEE is less than $0.2$. Moreover, we added $M$ phantom units known as ``sinks'' (\citealp{lu2001matching, baiocchi2010building}) to help eliminate surgeons who for no good match, i.e., another similar surgeon with distinct preference for TTE, can be found. Both $M$ and $\lambda$ are tuning parameters used to ensure that surgeons with markedly different TEE preference are paired together, while maintaining good covariate balance after matching. The matching algorithm formed matched pairs by minimizing the total distances between matched pairs of surgeons. After matching, no two clusters in one matched pair have their encouragement dosage less than $0.23$, while good covariate balance is maintained; see Table 1 in the main article.

\end{document}